\documentclass{article}
\usepackage[utf8]{inputenc}

\usepackage[english]{babel}

\usepackage[utf8]{inputenc}
\usepackage[T1]{fontenc}
\usepackage{scrextend}

\usepackage{amssymb}
\usepackage{stmaryrd}
\usepackage{amsmath,amsfonts}
\usepackage[tikz]{bclogo}
\usepackage[top=4.34cm]{geometry}
\usepackage{subcaption}

\usepackage{lmodern}
\usepackage{textcomp}
\usepackage{mathtools, bm}
\usepackage{amssymb, bm}
\usepackage[unq]{unique}
\usepackage{enumerate}
\usepackage{comment}
\usepackage[breaklinks]{hyperref}

\usepackage[numbers]{natbib}  

\usepackage[breaklinks]{hyperref}
\usepackage[numbers]{natbib}

\usepackage{amsthm}

\usepackage{cleveref}

\usepackage{authblk}
\usepackage{tikz}
\usepackage{caption}
\usepackage{subcaption}
\usetikzlibrary{decorations.pathreplacing,calligraphy,quotes}

\newtheorem{theorem}{Theorem}[section]
\newtheorem{lemma}[theorem]{Lemma}
\newtheorem{observation}[theorem]{Observation}
\newtheorem{conjecture}[theorem]{Conjecture}
\newtheorem{proposition}[theorem]{Proposition}

\newtheorem{problem}[theorem]{Problem}

\newtheorem{remark}[theorem]{Remark}
\newtheorem{construction}[theorem]{Construction}

\theoremstyle{definition}
\newtheorem{definition}[theorem]{Definition}

\usepackage{fullpage}

\newcommand{\blfd}[2]{$(#1,#2)$-BLFD}
\newcommand{\rej}[1]{$(\cF\cF,#1)$-rejector}
\newcommand{\dec}[2]{$(F_{#1}, F_{#2})$-decomposition}
\newcommand{\mdec}[1]{$({#1}, 1)$-BLF decomposition}
\newcommand{\sdec}[1]{$#1$-MBSFD}
\newcommand{\kor}[1]{$#1$-$OR$}
\newcommand{\mfor}[1]{$(\cM,#1)$-forcer}
\newcommand{\ffor}[2]{$(#1 \cF,#2)$-forcer}

\newcommand{\khalf}{\lfloor \frac{k}{2} \rfloor}
\newcommand{\true}{\text{TRUE}}
\newcommand{\false}{\text{FALSE}}
\newcommand{\la}{\text{la}}
\newcommand{\cM}{\mathcal{M}}
\newcommand{\cF}{\mathcal{F}}
\newcommand{\cC}{\mathcal{C}}
\newcommand{\cP}{\mathcal{P}}

\title{The complexity of decomposing a graph into a matching and a bounded linear forest}

\author{Agnijo Banerjee\thanks{\{ab2558, am2758, jp895, jp899\}@cam.ac.uk, Department of Pure Mathematics and Mathematical Statistics (DPMMS), University of Cambridge, Wilberforce Road, Cambridge, CB3 0WA, United Kingdom} , 
Jo\~ao Pedro Marciano\thanks{{joao.marciano@impa.br}, IMPA, Estrada Dona Castorina 110, Jardim Botânico, Rio de Janeiro, RJ, Brazil} , 
Adva Mond$^{*}$, 
Jan Petr$^{*}$,
Julien Portier$^{*}$}
\date{}

\begin{document}

\maketitle

\begin{abstract}
Deciding whether a graph can be edge-decomposed into a matching and a $k$-bounded linear forest was recently shown by Campbell, H{\"o}rsch and Moore to be NP-complete for every $k \ge 9$, and solvable in polynomial time for $k=1,2$.
In the first part of this paper, we close this gap by showing that this problem is in NP-complete for every $k \ge 3$. 
In the second part of the paper, we show that deciding whether a graph can be edge-decomposed into a matching and a $k$-bounded star forest is polynomially solvable for any $k \in \mathbb{N} \cup \{ \infty \}$, answering another question by Campbell, H{\"o}rsch and Moore from the same paper.
\end{abstract}

\section{Introduction}

Decomposing the edge set of a graph into parts which span simpler subgraphs is a central problem in graph theory.
Deciding whether the edge set of a graph can be decomposed into isomorphic copies of a given graph $H$ is known to be NP-complete whenever $H$ contains a connected component with $3$ edges or more~\cite{dor1992graph}.
However, when considering edge-decompositions into copies of subgraphs taken from a larger family, the problem sometimes becomes easier.
For instance, the celebrated Nash-Williams Theorem~\cite{nash1964decomposition} characterises when a graph can be decomposed into $k$ forests, and this can be checked in polynomial time.

Perhaps the simplest family to consider in this context is the family consisting of all matchings, as a decomposition of a graph into matchings is equivalent to a proper colouring of its edges.
An immediate generalisation of the family of matchings would be families of \emph{linear forests}, which are forests consisting of vertex-disjoint paths.
We say that a linear forest is \emph{$k$-bounded} if each of its components is a path consisting of at most $k$ edges.
In~\cite{campbell2023decompositions} Campbell, H{\"o}rsch and Moore were investigated decompositions of graphs into two linear forests of bounded sizes. 
More precisely, they studied the computational complexity of the following problem class:

\paragraph*{$(k,l)$-bounded linear forest decomposition problem (\blfd{k}{l}).}
Given a graph $G$, is there a decomposition of its edge set into a $k$-bounded linear forest and an $l$-bounded linear forest? \\

As it is trivially impossible to decompose a graph with maximum degree at least $4$ into a matching and a linear forest, the inputs can be restricted to subcubic graphs.

Before we discuss the complexity of \blfd{k}{l}, we briefly mention a few related existence results. Akiyama, Exoo and Harary~\cite{akiyama1980covering} studied the linear arboricity $\la(G)$ of a graph $G$, which is the minimum number of parts in an edge-decomposion of $G$ into linear forests, and posed the \emph{Linear Arboricity Conjecture} which states that $\la(G) \le \left\lceil \frac{\Delta(G)+1}{2} \right\rceil$ for every graph $G$.
This conjecture is still open, and the best upper bound known so far is due to Lang and Postle~\cite{lang2020improved} who proved that $\la(G) \le \frac{\Delta}{2} + O(\sqrt{\Delta} \log^4\Delta)$ where $\Delta=\Delta(G)$.
A main result in decomposing subcubic graphs into linear forests is by Thomassen~\cite{thomassen1999two}, stating that the edges of any cubic graph can be decomposed into two $5$-bounded linear forests.
We refer the reader to~\cite{balogh2005covering,gonccalves2009covering,kronenberg2022decomposing,wormaldconjecture} for further related and more recent results.
Furthermore, a detailed account of general results concerning edge decompositions of graphs can be found in a survey by Glock, K\"{u}hn and Osthus~\cite{glock2021extremal}.

 Transcripted with our notation, Péroche~\cite{peroche1984np} showed that
\blfd{\infty}{\infty} is NP-complete.
Bermond, Fouquet, Habib and P\'{e}roche~\cite{bermond1984linear} proved that \blfd{3}{3} is NP-complete and conjectured that \blfd{k}{k} is NP-complete for every $k \ge 2$.
Jiang, Jiang and Jiang~\cite{jiang2021decomposing} verified the case $k=2$, and later on Campbell, H\"{o}rsch and Moore~\cite{campbell2023decompositions} settled the conjecture for any $k\ge 2$.
In fact, they proved a more general result, showing that \blfd{k}{l} is NP-complete if $k,l \ge 2$, or $l=1$ and $k \geq 9$, and polynomial if $l=1$ and $k \le 2$.
This solves the general problem for most cases, leaving a gap of $3 \le k \le 8$ and $l=1$ of last unsolved cases.
In this paper we close this gap.

\begin{theorem}\label{thm:nonplanar}
    \blfd{k}{1} is NP-complete for all $k \geq 3$.
\end{theorem}

Connecting this theorem to previous results, we obtain that \blfd{k}{l} exhibits the following dichotomy.

\begin{theorem}
\label{thm:classification}
    \blfd{k}{l} is NP-complete if $k+l \ge 4$, and can be solved in polynomial time otherwise.
\end{theorem}

We derive \Cref{thm:nonplanar} from an even stronger result. We define the problem \emph{Planar \blfd{k}{l}} to be the problem \blfd{k}{l} when the input is restriced to planar graphs $G$, about which we prove the following result.

\begin{theorem}
\label{thm:main}
    Planar \blfd{k}{1} is NP-complete for all $k \ge 3$.
\end{theorem}

Clearly, \Cref{thm:main} implies \Cref{thm:nonplanar}.

Another question raised in~\cite{campbell2023decompositions} is whether one could prove similar results when considering a matching and a $k$-bounded star forest instead of a $k$-bounded linear forest.

\paragraph*{Matching and $k$-bounded star forest decomposition problem (\sdec{k}).}
Given a graph $G$, is there a decomposition of its edge set into a matching and a $k$-bounded star forest? \\

The case $k=2$ is clearly equivalent to the \blfd{2}{1} problem, which was shown to be polynomial time solvable in \cite{campbell2023decompositions}. 
By a direct generalisation of their proof, we answer Campbell, H{\"o}rsch and Moore's question.

\begin{theorem}
\label{thm:DecompositionMatchingStarForest}
   $k$-MBSFD is solvable in polynomial time for every $k \in \mathbb{N} \cup \{\infty \}$.
\end{theorem}

The rest of this paper is organised as follows. In Section \ref{sec:NPComplete} we prove \Cref{thm:main}, starting with a gentle outline of the proof. In Section \ref{sec:StarForestinP} we prove \Cref{thm:DecompositionMatchingStarForest} and in Section \ref{sec:ConcludingRemarks} we state some open problems.

\section{NP-completeness of \blfd{k}{1} for $k \ge 3$}\label{sec:NPComplete}

\subsection{Proof outline}

We firstly note that there exists a certificate checking in polynomial time whether, given an integer $k$, a graph $G = (V,E)$ and two subsets $F_k,M\subset E$, the pair $(F_k, M)$ form a partition of $E$ such that $F_k$ is a $k$-bounded linear forest and $M$ is a matching.
Hence, in order to prove that Planar \blfd{k}{1} is NP-complete for any $k\ge 3$, it is sufficient to prove that it is NP-hard.
This then implies \Cref{thm:main}.
For this purpose, we show a reduction of a variant of $3$-SAT to \blfd{k}{1}.

For a variable $x$ we denote by $\bar{x}$ its negative literal.
By a \emph{\mdec{k}} $(F_k,M)$ of a graph $G$ we mean an edge-decomposition into a $k$-bounded linear forest $F_k$ and a matching $M$.
A path of length $\ell$ (in edges) will also be referred to as an \emph{$\ell$-path}.
We consider the following SAT problem.

\subsubsection*{$(\le 3, 3)$–SAT}
\textbf{Input:} 
A set of variables $X$, a set of clauses $\cC$ such that $|C| \in \{2,3\}$ for all $C \in \cC$, and for every $x \in X$ we have that $x$ appears in exactly two clauses and $\bar{x}$ appears in exactly one clause. \\ 
\textbf{Question:} Is there a truth assignment $\phi : X \rightarrow \left\{\text{TRUE}, \text{FALSE} \right\}$ such that every clause contains at least one true literal?

An assignment $\phi$ with the desired properties is said to be \emph{satisfying}.\\

Since we wish to prove the NP-completeness of Planar \blfd{k}{1}, we consider a slightly more restricted variant of $(\le 3, 3)$-SAT.
Given a tuple $(X,\cC)$ where $X$ is a set of variables and $\cC$ is a set of clauses, its \emph{incidence graph} is the bipartite graph with parts $X \cup \cC$, where $xC$ is an edge if and only if either $x$ or $\bar{x}$ appears in $C$.
\emph{Planar $(\le 3,3)$-SAT} is $(\le 3, 3)$-SAT restricted to instances which have planar incidence graphs.

\begin{theorem}[\cite{dahlhaus1994complexity,middendorf1993complexity}]
\label{thm:NPcomplete3bounded3SAT}
    Planar $(\le 3,3)$–SAT is NP-complete.
\end{theorem}

Our reduction from a variant of $3$-SAT to Planar \blfd{k}{1} follows the general lines of the one in~\cite{campbell2023decompositions}.
Given an instance $(X,\cC)$ of Planar $(\le 3,3)$-SAT, we construct a planar graph $G_{(X,\cC)}$, and we show that $(X,\cC)$ is a yes-instance of Planar $(\le 3, 3)$-SAT if and only if $G_{(X,\cC)}$ is a yes-instance of Planar \blfd{k}{1}.
Clauses in $\cC$ are represented in our graph as vertices, and variables are represented as subgraphs which we call \emph{variable gadgets}.
For each variable $x \in X$ and a clause $C \in \cC$ such that either $x$ or $\bar{x}$ appears in $C$, we put an edge between the gadget variable representing $x$ and the vertex representing $C$.
We do this such that if $x$ appears in $C$ then the gadget representing $x$ will be connected to the vertex $C$ through a vertex which we call a \emph{positive input vertex}, and through a vertex which we call a \emph{negative input vertex} if $\bar{x}$ appears in $C$.
As $(X,\cC)$ is an instance of $(\le 3, 3)$-SAT, each variable gadget in $G_{(X,\cC)}$ will have two positive input vertices and one negative input vertex (see \Cref{fig:general}).

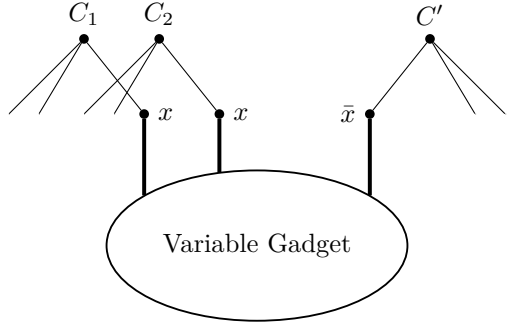
\begin{figure}
        \centering	\begin{tikzpicture}[auto, vertex/.style={circle,draw=black!100,fill=black!100, thick,inner sep=0pt,minimum size=1mm}, novertex/.style={circle,draw=black!100,fill=black!100,inner sep=0pt,minimum size=0mm}]]
		\node (p1) at (-1.5,1.75) [vertex,label=right:$x$] {};
            \node (p2) at (-0.5,1.75) [vertex,label=right:$x$] {};
            \node (j1) at (1.5,1.75) [vertex,label=left:$\bar{x}$] {};
            \node (v1) at (-1.5,0.6614) [novertex] {};
            \node (v2) at (-0.5,0.9682) [novertex] {};
            \node (u1) at (1.5,0.6614) [novertex] {};

            \node (c1) at (-2.3,2.75) [vertex,label=$C_1$] {};
            \node (c2) at (-1.3,2.75) [vertex,label=$C_2$] {};
            \node (c3) at (2.3,2.75) [vertex,label=$C'$] {};

            \node (c11) at (-2.9,1.75) [novertex] {};
            \node (c12) at (-3.3,1.75) [novertex] {};
            \node (c21) at (-1.9,1.75) [novertex] {};
            \node (c22) at (-2.3,1.75) [novertex] {};
            \node (c31) at (2.9,1.75) [novertex] {};
            \node (c32) at (3.3,1.75) [novertex] {};
            
            \draw [line width=0.25mm] (0,0) ellipse (2cm and 1cm);
		\node [label={[shift={(0,-0.4)}]Variable Gadget}] at (0,0) {};

            \draw [-,line width=1.5pt] (p1) --node[inner sep=2pt]{} (v1);
            \draw [-,line width=1.5pt] (p2) --node[inner sep=2pt]{} (v2);
            \draw [-,line width=1.5pt] (j1) --node[inner sep=2pt]{} (u1);

            \draw [-] (p1) --node[inner sep=2pt]{} (c1);
            \draw [-] (p2) --node[inner sep=2pt]{} (c2);
            \draw [-] (j1) --node[inner sep=2pt]{} (c3);

            \draw [-] (c1) --node[inner sep=2pt]{} (c11);
            \draw [-] (c1) --node[inner sep=2pt]{} (c12);
            \draw [-] (c2) --node[inner sep=2pt]{} (c21);
            \draw [-] (c2) --node[inner sep=2pt]{} (c22);
            \draw [-] (c3) --node[inner sep=2pt]{} (c31);
            \draw [-] (c3) --node[inner sep=2pt]{} (c32);
	\end{tikzpicture}
	\caption{An illustration of one variable gadget for a variable $x$, with two positive input vrtices on the left-hand-side and one negative input vertex on the right-hand-side, and of three clause vertices $C_1, C_2$ and $C'$, such that $x$ appears in $C_1, C_2$ and $\bar{x}$ appears in $C'$. Each clause vertex is adjacent to two more edges, each connecting it to a relevant variable gadget.}
    \label{fig:general}
\end{figure}

We then show the following correspondence between satisfying assignments of $(X,\cC)$ and \mdec{k}s of $G_{(X,\cC)}$.
Given a \mdec{k} of $G_{(X,\cC)}$, we will associate an edge between a gadget representing the variable $x$ and a vertex representing the clause $C$ being contained in the matching with $\phi(x) = \true$ if this edge is adjacent to a positive input vertex, and with $\phi(x) = \false$ if the edge is adjacent to a negative input vertex.
This is meant to represent that the assignment of $\phi$ to $x$ makes $C$ true.
For this heuristic to make sense, one must exclude the possibility of a \mdec{k} from which both $\phi(x)=\true$ and $\phi(x)=\false$ are deduced.
This corresponds to what we call \emph{clashing edges}, that is, having a \mdec{k} of our graph such that there exists a variable $x$ with both the negative input vertex and at least one of the positive input vertices being adjacent to a clause by an edge which is contained in the matching.

Given this framework, as long as we construct a variable gadget which does not allow clashing edges in any \mdec{k} of it, then it must be possible to deduce a satisfying assignment $\phi$ for $(X,\cC)$ from any \mdec{k} of $G_{(X,\cC)}$.
Indeed, we then have that each variable will be assigned exactly one truth value, and each clause vertex will have at least one variable which its assignment makes it true.
Conversely, we must be able to deduce a \mdec{k} of $G_{(X,\cC)}$ from a satisfying assignment $\phi$ of $(X,\cC)$.
Note that in our framework an assigment $\phi$ induces a \mdec{k} of all the edges adjacent to clause vertices.
Hence, this should be possible as long as any \mdec{k} of the edges adjacent to clause vertices can be ``extended'' to valid \mdec{k}s of each variable gadget, such that we obtain a valid \mdec{k} of the whole graph.

The challenge is then to construct a variable gadget which satisfies all the above.
Note that the first restriction on the variable gadget implies that having both its negative input vertex and one of its positive input vertices adjacent to a matching edge in a \mdec{k} at the same time must be impossible.
The second restriction means that as long as both its negative input vertex and one of its positive input vertices are not adjacent to a matching edge in a \mdec{k} of the edges adjacent to clause vertices, then one can extend it to a \mdec{k} of the whole variable gadget.

We construct our variable gadget by combining two auxiliary gadgets, which we call \kor{k} and \rej{k} (see \Cref{def:kOR,def:kFrejector}, respectively).
Informally, given a \mdec{k} of the graph $G_{(X,\cC)}$, the \kor{k} has the property of identifying when the variable gadget has a matching edge connecting it to a clause through one of its positive input vertices, and if this is the case, the \rej{k} has the property of rejecting the possibility of the edge connecting the negative input vertex to a clause vertex to also being contained in the matching.

The rest of this section is organised as follows.
In \Cref{sec:AuxGadgets} we construct the auxiliary gadgets, which are later used in \Cref{sec:VarGadget} to construct the variable gadget.
Then, in \Cref{sec:Reduction} we prove the reduction.

\subsection{Auxiliary gadgets}
\label{sec:AuxGadgets}

An important object we use to construct our gadgets are \emph{forcers}.
Since forcers are basic ``building blocks'' in our gadgets, we start by describing their construction.

\subsubsection{Forcer graphs}
$\cM$-forcers and $\cF_k$-forcers are graphs with one designated edge which is forced to be respectively in $M$ or in $F_k$, for any \mdec{k} $(F_k,M)$ of them.
Forcer graphs for $k \ge 4$ have already been constructed in~\cite{campbell2023decompositions}, but using graphs with multiple edges.
Here we construct forcers using simple graphs for every $k \geq 3$ for the sake of generalisation.

\begin{definition}
\label{def:Mforcer}
    Let $k \ge 3$ be an integer.
    Given a graph $G = (V,E)$ and a vertex $v \in V$ of degree $1$, we say that $(G,v)$ is an \emph{\mfor{k}} if the following holds.
    \begin{itemize}
        \item There exists an \mdec{k} of $G$.
        \item For any \mdec{k} $(F_k,M)$ of $G$, the edge $e \in E$ containing $v$ is in $M$.
    \end{itemize}
\end{definition}

\begin{definition}
\label{def:ellFforcer}
    Let $k \ge 3$ and $1 \le \ell \le k$ be integers.
    Given a graph $G = (V,E)$ and a vertex $v \in V$ of degree $1$, we say that $(G,v)$ is an \emph{\ffor{\ell}{k}} if the following holds.
    \begin{itemize}
        \item There exists an \mdec{k} of $G$.
        \item For any \mdec{k} $(F_k,M)$ of $G$, the vertex $v$ is contained in an $\ell$-path in $F_k$, and in no $(\ell+1)$-path in $F_k$. 
    \end{itemize} 
\end{definition}

We show that \mfor{k}s and \ffor{\ell}{k}s exist, by giving explicit constructions.

\begin{proposition}
\label{prop:Mforcer}
    For any $k \ge 3$ there exists an \mfor{k} $(G,v)$.
    Moreover, the statement still holds with the additional restriction of $G$ being a planar graph.
\end{proposition}

\begin{proof}
    We give an explicit construction of an \mfor{k} for any possible value of $k$.
    For $k \ge 8$, we construct $(G,v)$ as in \Cref{fig:Mforcer}.
    For $k \le 7$ we have a different construction, which splits further into cases where $k=3$, $k=4,5$, or $k=6,7$, as given in \Cref{fig:MforcerSmallk}.
    An easy verification shows that each of the given constructions is indeed an \mfor{k}.
    Note that in all constructions the graphs are planar.
\end{proof}

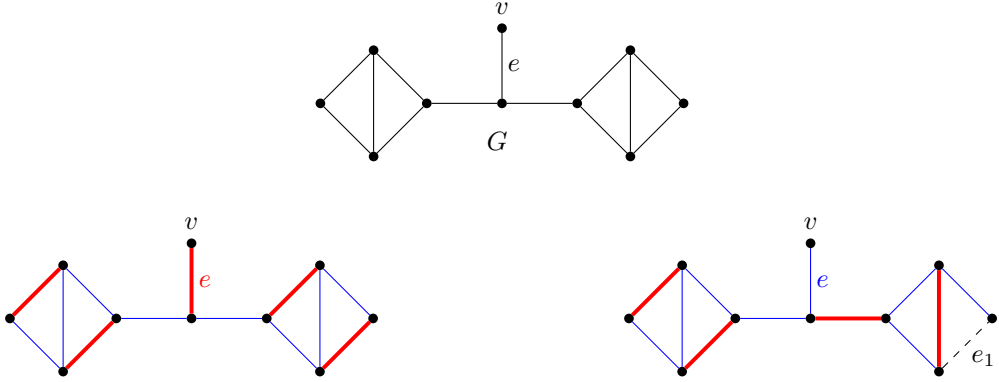
\begin{figure}
	\centering
	\begin{subfigure}[b]{0.5\textwidth}
        \centering
	\begin{tikzpicture}[auto, vertex/.style={circle,draw=black!100,fill=black!100, thick,inner sep=0pt,minimum size=1mm}]
			\node (x) at (0,0) [vertex] {};
			\node (y) at (0,1) [vertex,label=$v$] {};
			\node (v1) at (1,0) [vertex] {};
			\node (v2) at (1.707,0.707) [vertex] {};
			\node (v3) at (1.707,-0.707) [vertex] {};
			\node (v4) at (2.414,0) [vertex] {};
			\node (u1) at (-1,0) [vertex] {};
			\node (u2) at (-1.707,0.707) [vertex] {};
			\node (u3) at (-1.707,-0.707) [vertex] {};
			\node (u4) at (-2.414,0) [vertex] {};
			\node[text width=0.4cm] at (0,-0.5) {$G$};
			
			\draw [-] (x) --node[inner sep=2pt,swap]{$e$} (y);
			\draw [-] (x) --node[inner sep=2pt,swap]{} (u1);
			\draw [-] (x) --node[inner sep=2pt,swap]{} (v1);
			\draw [-] (v1) --node[inner sep=2pt,swap]{} (v2);
			\draw [-] (v1) --node[inner sep=2pt,swap]{} (v3);
			\draw [-] (v2) --node[inner sep=2pt,swap]{} (v3);
			\draw [-] (v2) --node[inner sep=2pt,swap]{} (v4);
			\draw [-] (v3) --node[inner sep=2pt,swap]{} (v4);
			\draw [-] (u1) --node[inner sep=2pt,swap]{} (u2);
			\draw [-] (u1) --node[inner sep=2pt,swap]{} (u3);
			\draw [-] (u2) --node[inner sep=2pt,swap]{} (u3);
			\draw [-] (u2) --node[inner sep=2pt,swap]{} (u4);
			\draw [-] (u3) --node[inner sep=2pt,swap]{} (u4);
		\end{tikzpicture}
	\end{subfigure}%
	\vspace{6mm} %
	\begin{subfigure}[b]{1\textwidth}
	\begin{minipage}{.5\textwidth}
		\centering
		\begin{tikzpicture}[auto, vertex/.style={circle,draw=black!100,fill=black!100, thick,inner sep=0pt,minimum size=1mm}]
			\node (x) at (0,0) [vertex] {};
			\node (y) at (0,1) [vertex,label=$v$] {};
			\node (v1) at (1,0) [vertex] {};
			\node (v2) at (1.707,0.707) [vertex] {};
			\node (v3) at (1.707,-0.707) [vertex] {};
			\node (v4) at (2.414,0) [vertex] {};
			\node (u1) at (-1,0) [vertex] {};
			\node (u2) at (-1.707,0.707) [vertex] {};
			\node (u3) at (-1.707,-0.707) [vertex] {};
			\node (u4) at (-2.414,0) [vertex] {};
			
			\draw [-,red,line width=1.5pt] (x) --node[inner sep=2pt,swap]{$e$} (y);
			\draw [-,blue] (x) --node[inner sep=2pt,swap]{} (u1);
			\draw [-,blue] (x) --node[inner sep=2pt,swap]{} (v1);
			\draw [-,red,line width=1.5pt] (v1) --node[inner sep=2pt,swap]{} (v2);
			\draw [-,blue] (v1) --node[inner sep=2pt,swap]{} (v3);
			\draw [-,blue] (v2) --node[inner sep=2pt,swap]{} (v3);
			\draw [-,blue] (v2) --node[inner sep=2pt,swap]{} (v4);
			\draw [-,red,line width=1.5pt] (v3) --node[inner sep=2pt,swap]{} (v4);
			\draw [-,blue] (u1) --node[inner sep=2pt,swap]{} (u2);
			\draw [-,red,line width=1.5pt] (u1) --node[inner sep=2pt,swap]{} (u3);
			\draw [-,blue] (u2) --node[inner sep=2pt,swap]{} (u3);
			\draw [-,red,line width=1.5pt] (u2) --node[inner sep=2pt,swap]{} (u4);
			\draw [-,blue] (u3) --node[inner sep=2pt,swap]{} (u4);
		\end{tikzpicture}
	\end{minipage}%
	\begin{minipage}{.5\textwidth}
		\centering
		\begin{tikzpicture}[auto, vertex/.style={circle,draw=black!100,fill=black!100, thick,inner sep=0pt,minimum size=1mm}]
			\node (x) at (0,0) [vertex] {};
			\node (y) at (0,1) [vertex,label=$v$] {};
			\node (v1) at (1,0) [vertex] {};
			\node (v2) at (1.707,0.707) [vertex] {};
			\node (v3) at (1.707,-0.707) [vertex] {};
			\node (v4) at (2.414,0) [vertex] {};
			\node (u1) at (-1,0) [vertex] {};
			\node (u2) at (-1.707,0.707) [vertex] {};
			\node (u3) at (-1.707,-0.707) [vertex] {};
			\node (u4) at (-2.414,0) [vertex] {};
			
			\draw [-,blue] (x) --node[inner sep=2pt,swap]{$e$} (y);
			\draw [-,blue] (x) --node[inner sep=2pt,swap]{} (u1);
			\draw [-,red,line width=1.5pt] (x) --node[inner sep=2pt,swap]{} (v1);
			\draw [-,blue] (v1) --node[inner sep=2pt,swap]{} (v2);
			\draw [-,blue] (v1) --node[inner sep=2pt,swap]{} (v3);
			\draw [-,red,line width=1.5pt] (v2) --node[inner sep=2pt,swap]{} (v3);
			\draw [-,blue] (v2) --node[inner sep=2pt,swap]{} (v4);
			\draw [-,dashed] (v3) --node[inner sep=2pt,swap]{$e_1$} (v4);
			\draw [-,blue] (u1) --node[inner sep=2pt,swap]{} (u2);
			\draw [-,red,line width=1.5pt] (u1) --node[inner sep=2pt,swap]{} (u3);
			\draw [-,blue] (u2) --node[inner sep=2pt,swap]{} (u3);
			\draw [-,red,line width=1.5pt] (u2) --node[inner sep=2pt,swap]{} (u4);
			\draw [-,blue] (u3) --node[inner sep=2pt,swap]{} (u4);
		\end{tikzpicture}
	\end{minipage}
        \end{subfigure}
	\caption{An \mfor{k} $(G,v)$ for $k \ge 8$ on the top, on the left-hand-side a possible \mdec{k} of it $(F_k,M)$ where $e \in M$, and on the right-hand-side an attempt for a \mdec{k} $(F_k,M)$ where $e \in F_k$ creating an edge $e_1$ with no possible assignment. $F_k$ is coloured blue and $M$ is in bold and coloured red.}
	\label{fig:Mforcer}
\end{figure}

\begin{figure}
\centering
\begin{subfloat}[$k=3$] {
\centering
    \begin{tikzpicture}[auto, vertex/.style={circle,draw=black!100,fill=black!100, thick,inner sep=0pt,minimum size=1mm}]
        \node (x) at (0,0) [vertex] {};
        \node (y) at (0,1) [vertex] {};
        \node (v) at (0,2) [vertex,label=$v$] {};
        \node (v1) at (1,0) [vertex] {};
        \node (v2) at (1.707,0.707) [vertex] {};
        \node (v3) at (1.707,-0.707) [vertex] {};
        \node (u1) at (-1,0) [vertex] {};
        \node (u2) at (-1.707,0.707) [vertex] {};
        \node (u3) at (-1.707,-0.707) [vertex] {};
        \node[text width=0.4cm] at (0,-0.5) {$G$};

        \draw [-,red,line width=1.5pt] (y) --node[inner sep=2pt,swap]{$e$} (v);
        \draw [-,blue] (x) --node[inner sep=2pt]{} (y);
        \draw [-,blue] (x) --node[inner sep=2pt]{} (u1);
        \draw [-,red,line width=1.5pt] (x) --node[inner sep=2pt]{} (v1);
        \draw [-,blue] (v1) --node[inner sep=2pt]{} (v2);
        \draw [-,blue] (v1) --node[inner sep=2pt]{} (v3);
        \draw [-,blue] (u1) --node[inner sep=2pt]{} (u2);
        \draw [-,red,line width=1.5pt] (u1) --node[inner sep=2pt]{} (u3);
    \end{tikzpicture}
}
\end{subfloat}
\hspace{10pt}       
\begin{subfloat}[$k=4,5$] {
\centering
    \begin{tikzpicture}[auto, vertex/.style={circle,draw=black!100,fill=black!100, thick,inner sep=0pt,minimum size=1mm}]
	\node (x) at (0,0) [vertex] {};
	\node (v) at (0,1) [vertex,label=$v$] {};
        \node (xv) at (1,0) [vertex] {};
        \node (yv) at (1.707,0.707) [vertex] {};
        \node (xu) at (-1,0) [vertex] {};
        \node (yu) at (-1.707,0.707) [vertex] {};
        \node (v1) at (1.707,-0.707) [vertex] {};
        \node (v2) at (2.414,-1.414) [vertex] {};
        \node (v3) at (1,-1.414) [vertex] {};
        \node (v4) at (2.414,-2.414) [vertex] {};
        \node (v5) at (1,-2.414) [vertex] {};
        \node (u1) at (-1.707,-0.707) [vertex] {};
        \node (u2) at (-2.414,-1.414) [vertex] {};
        \node (u3) at (-1,-1.414) [vertex] {};
        \node (u4) at (-2.414,-2.414) [vertex] {};
        \node (u5) at (-1,-2.414) [vertex] {};
        \node[text width=0.4cm] at (0,-0.5) {$G$};
			
        \draw [-,red,line width=1.5pt] (x) --node[inner sep=2pt,swap]{$e$} (v);
        \draw [-,blue] (x) --node[inner sep=2pt]{} (xv);
        \draw [-,blue] (xv) --node[inner sep=2pt]{} (yv);
        \draw [-,red,line width=1.5pt] (xv) --node[inner sep=2pt]{} (v1);
        \draw [-,blue] (v1) --node[inner sep=2pt]{} (v2);
        \draw [-,blue] (v1) --node[inner sep=2pt]{} (v3);
        \draw [-,red,line width=1.5pt] (v2) --node[inner sep=2pt]{} (v3);
        \draw [-,blue] (v2) --node[inner sep=2pt]{} (v4);
        \draw [-,blue] (v3) --node[inner sep=2pt]{} (v5);
        \draw [-,red,line width=1.5pt] (v4) --node[inner sep=2pt]{} (v5);
        \draw [-,blue] (x) --node[inner sep=2pt]{} (xu);
        \draw [-,blue] (xu) --node[inner sep=2pt]{} (yu);
        \draw [-,red,line width=1.5pt] (xu) --node[inner sep=2pt]{} (u1);
        \draw [-,blue] (u1) --node[inner sep=2pt]{} (u2);
        \draw [-,blue] (u1) --node[inner sep=2pt]{} (u3);
        \draw [-,red,line width=1.5pt] (u2) --node[inner sep=2pt]{} (u3);
        \draw [-,blue] (u2) --node[inner sep=2pt]{} (u4);
        \draw [-,blue] (u3) --node[inner sep=2pt]{} (u5);
        \draw [-,red,line width=1.5pt] (u4) --node[inner sep=2pt]{} (u5);
    \end{tikzpicture}
}       
\end{subfloat}
\vspace{20pt}%
\begin{subfloat}[$k=6,7$] {
\centering
    \begin{tikzpicture}[auto, vertex/.style={circle,draw=black!100,fill=black!100, thick,inner sep=0pt,minimum size=1mm}]
        \node (x) at (0,0) [vertex] {};
        \node (v) at (0,1) [vertex,label=$v$] {};
        \node (xv) at (1,0) [vertex] {};
        \node (yv) at (1.707,0.707) [vertex] {};
        \node (xu) at (-1,0) [vertex] {};
        \node (yu) at (-1.707,0.707) [vertex] {};
        \node (v1) at (1.707,-0.707) [vertex] {};
        \node (v2) at (2.414,-1.414) [vertex] {};
        \node (v3) at (1,-1.414) [vertex] {};
        \node (v4) at (2.414,-2.414) [vertex] {};
        \node (v5) at (1,-2.414) [vertex] {};
        \node (v6) at (2.414,-3.414) [vertex] {};
        \node (v7) at (1,-3.414) [vertex] {};
        \node (u1) at (-1.707,-0.707) [vertex] {};
        \node (u2) at (-2.414,-1.414) [vertex] {};
        \node (u3) at (-1,-1.414) [vertex] {};
        \node (u4) at (-2.414,-2.414) [vertex] {};
        \node (u5) at (-1,-2.414) [vertex] {};
        \node (u6) at (-2.414,-3.414) [vertex] {};
        \node (u7) at (-1,-3.414) [vertex] {};
        \node[text width=0.4cm] at (0,-0.5) {$G$};
        
        \draw [-,red,line width=1.5pt] (x) --node[inner sep=2pt,swap]{$e$} (v);
        \draw [-,blue] (x) --node[inner sep=2pt]{} (xv);
        \draw [-,blue] (xv) --node[inner sep=2pt]{} (yv);
        \draw [-,red,line width=1.5pt] (xv) --node[inner sep=2pt]{} (v1);
        \draw [-,blue] (v1) --node[inner sep=2pt]{} (v2);
        \draw [-,blue] (v1) --node[inner sep=2pt]{} (v3);
        \draw [-,red,line width=1.5pt] (v2) --node[inner sep=2pt]{} (v3);
        \draw [-,blue] (v2) --node[inner sep=2pt]{} (v4);
        \draw [-,blue] (v3) --node[inner sep=2pt]{} (v5);
        \draw [-,red,line width=1.5pt] (v4) --node[inner sep=2pt]{} (v5);
        \draw [-,blue] (v4) --node[inner sep=2pt]{} (v6);
        \draw [-,blue] (v5) --node[inner sep=2pt]{} (v7);
        \draw [-,red,line width=1.5pt] (v6) --node[inner sep=2pt]{} (v7);
        \draw [-,blue] (x) --node[inner sep=2pt]{} (xu);
        \draw [-,blue] (xu) --node[inner sep=2pt]{} (yu);
        \draw [-,red,line width=1.5pt] (xu) --node[inner sep=2pt]{} (u1);
        \draw [-,blue] (u1) --node[inner sep=2pt]{} (u2);
        \draw [-,blue] (u1) --node[inner sep=2pt]{} (u3);
        \draw [-,red,line width=1.5pt] (u2) --node[inner sep=2pt]{} (u3);
        \draw [-,blue] (u2) --node[inner sep=2pt]{} (u4);
        \draw [-,blue] (u3) --node[inner sep=2pt]{} (u5);
        \draw [-,red,line width=1.5pt] (u4) --node[inner sep=2pt]{} (u5);
        \draw [-,blue] (u4) --node[inner sep=2pt]{} (u6);
        \draw [-,blue] (u5) --node[inner sep=2pt]{} (u7);
        \draw [-,red,line width=1.5pt] (u6) --node[inner sep=2pt]{} (u7);
    \end{tikzpicture}
}       
\end{subfloat}
\caption{\mfor{k}s $(G,v)$ for $3 \le k \le 7$ with a possible \mdec{k} $(F_k,M)$ of each where $e \in M$. $F_k$ is coloured blue and $M$ is in bold and coloured red.}
\label{fig:MforcerSmallk}
\end{figure}
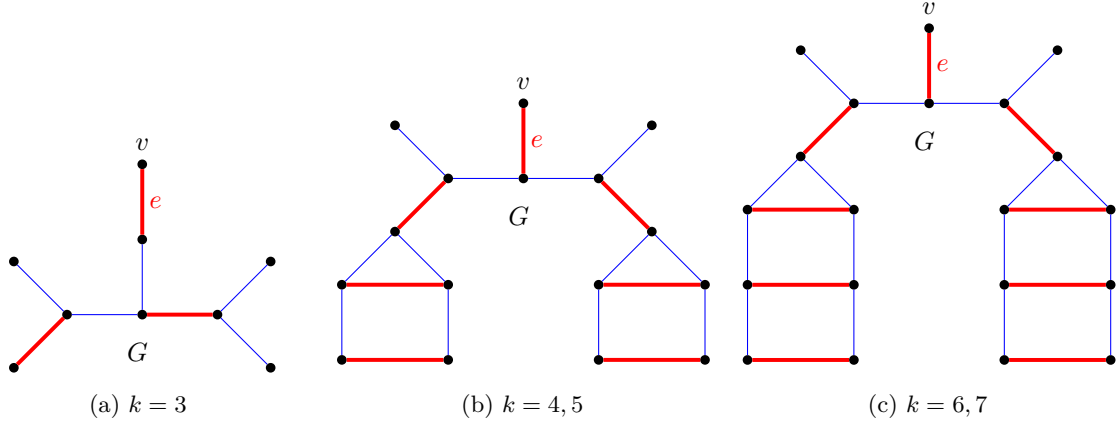

\begin{proposition}
\label{prop:ellFforcer}
    For any integers $k \ge 3$ and $1 \le \ell \le k$ there exists an \ffor{\ell}{k} $(G,v)$.
    Moreover, the statement still holds with the additional restriction of $G$ being a planar graph.
\end{proposition}

\begin{proof}
    Let $k \ge 3$ be an integer and consider first the case where $\ell = 1$.
    Let $(G',v')$ be the \mfor{k} as constructed in the proof of \Cref{prop:Mforcer}.
    Let $G_1$ be the graph obtained by adding a new vertex $v$ and the edge $e = vv'$.
    Then $(G_1,v)$ is a \ffor{1}{k}.
    Indeed, any \mdec{k} $(F_k,M)$ of $G_1$ induces an \mdec{k} $(F'_k,M')$ of $G'$. 
    Since $(G',v')$ is an \mfor{k}, we get that $e' \in M'$, where $e' \in E(G')$ is the edge satisfying $v' \in e'$.  
    It then follows that in $(F_k,M)$ we have $e=vv' \in F_k$, and moreover, that $v$ is not contained in any $2$-path in $F_k$ (see \Cref{fig:ellFforcer}(a)).
    
    Now let $2 \le \ell \le k$ be an integer, and let $(G^{(1)}, v_1), \ldots, (G^{(\ell-1)}, v_{\ell-1})$ be $\ell-1$ disjoint copies of $(G',v')$.
    Define the graph $G_{\ell} = (V,E)$ as follows.
    Let $v, u$ be two vertices not in $V(G^{(i)})$ for all $i=1, \ldots, \ell-1$, and let $P$ be the $\ell$-path $(v, v_1, \ldots, v_{\ell-1}, u)$.
    Set $V \coloneqq V(G^{(1)}) \cup \ldots \cup V(G^{(\ell-1)}) \cup \{v, u \}$ and $E \coloneqq E(G^{(1)}) \cup \ldots \cup E(G^{(\ell-1)}) \cup E(P)$ (see \Cref{fig:ellFforcer}(b)).
    We claim that $(G_{\ell},v)$ is an \ffor{\ell}{k}.
    Indeed, any \mdec{k} $(F_k,M)$ of $G_{\ell}$ induces \mdec{k}s $(F^{(i)}_k,M^{(i)})$ of $G^{(i)}$ for $i = 1, \ldots, \ell-1$.
    Since the $(G^{(i)}, v_i)$'s are \mfor{k}s, we know that we must have $e_i \in M^{(i)}$, where $e_i \in E(G^{(i)})$ is the edge with $v_i \in e_i$, for all $i=1, \ldots, \ell-1$.
    Hence, in $(F_k,M)$ we must have that $v v_1, v_{\ell -1}u \in F_k$ and $ v_iv_{i+1} \in F_k$ for all $i=1,\ldots, \ell-2$, implying that $P \subset F_k$, and moreover, no $(\ell+1)$-path containing $v$ is in $F_k$.
    
    Lastly, the existence of a \mdec{k} of $G_{\ell}$, for any value of $\ell$, is demonstrated in \Cref{fig:ellFforcer} by explicit decompositions.
\end{proof}

\begin{figure}
\centering
\begin{subfloat}[$\ell=1$] {
\centering
    \begin{tikzpicture}[auto, vertex/.style={circle,draw=black!100,fill=black!100, thick,inner sep=0pt,minimum size=1mm}]
        \node (v) at (2,0) [vertex,label=$v$] {};
        \node (v') at (1,0) [vertex,label=$v'$] {};
        \node (G') at (-0.54,0) [circle,draw=red,inner ysep=1.5em] {\color{red}$G'$};
        \node (G) at (1,-0.7) [] {$G_1$};

        \draw [-,blue] (v) --node[inner sep=2pt,swap]{$e$} (v');
        \draw [-,red,line width=1.5pt] (v') --node[inner sep=2pt,swap]{} (G');
    \end{tikzpicture}
}
\end{subfloat}
\hspace{10pt}       
\begin{subfloat}[$\ell \ge 2$] {
\centering
    \begin{tikzpicture}[auto, vertex/.style={circle,draw=black!100,fill=black!100, thick,inner sep=0pt,minimum size=1mm}]
        \node (v0) at (-3.6,0) [vertex,label=$v_0$] {};
        \node (v1) at (-2.3,0) [vertex,label=$v_1$] {};
        \node (v2) at (-1,0) [vertex,label=$v_2$] {};
        \node (vell-2) at (1,0) [vertex,label=$v_{\ell-2}$] {};
        \node (vell-1) at (2.3,0) [vertex,label=$v_{\ell-1}$] {};
        \node (vell) at (3.6,0) [vertex,label=$v_{\ell}$] {};
        \node (G1) at (-2.3,-1.4) [circle,draw=red,inner ysep=0.8em] {\color{red}\small{$G^{(1)}$}};
        \node (G2) at (-1,-1.4) [circle,draw=red,inner ysep=0.8em] {\color{red}\small{$G^{(2)}$}};
        \node (Gell-2) at (1,-1.4) [circle,draw=red,inner ysep=0em] {\color{red}\small{$G^{(\ell-2)}$}};
        \node (Gell-1) at (2.3,-1.4) [circle,draw=red,inner ysep=0em] {\color{red}\small{$G^{(\ell-1)}$}};
        \node (G) at (0,-0.7) [] {$G_{\ell}$};
        
        \draw [-,blue] (v0) --node[inner sep=2pt,swap]{} (v1);
        \draw [-,blue] (v1) --node[inner sep=2pt,swap]{} (v2);
        \draw [-,dotted,blue] (v2) --node[inner sep=2pt,swap]{} (vell-2);
        \draw [-,blue] (vell-2) --node[inner sep=2pt,swap]{} (vell-1);
        \draw [-,blue] (vell-1) --node[inner sep=2pt,swap]{} (vell);
        \draw [-,red,line width=1.5pt] (v1) --node[inner sep=2pt,swap]{} (G1);
        \draw [-,red,line width=1.5pt] (v2) --node[inner sep=2pt,swap]{} (G2);
        \draw [-,red,line width=1.5pt] (vell-2) --node[inner sep=2pt,swap]{} (Gell-2);
        \draw [-,red,line width=1.5pt] (vell-1) --node[inner sep=2pt,swap]{} (Gell-1);
    \end{tikzpicture}
}       
\end{subfloat}
\caption{\ffor{\ell}{k}s $(G,v)$ with a possible \mdec{k} $(F_k,M)$ for each satisfying the conditions. $F_k$ is coloured blue and the \mfor{k}s' graphs $G', G^{(1)}, \ldots, G^{(\ell-1)}$ are coloured red, as well as $M$, which is also in bold.}
\label{fig:ellFforcer}
\end{figure}
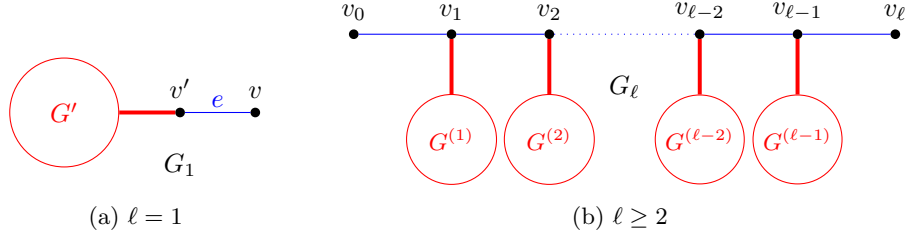

\subsubsection{\kor{k} gadget}

The first gadget we build using a forcer is a \kor{k} gadget. It is a gadget with three designated edges satisfying that for any \dec{k}{l} $(F_k,M)$ of its graph, the third edge belongs to $F_k$ if and only if at least one of the first two edges belongs to $F_k$. As we shall be associating edges belonging to $F_k$ with \true, the gadget serves as a disjunction.

\begin{definition}
\label{def:kOR}
    Let $k \geq 3$.
    Let $O=(V,E)$ be a graph and $p_1,p_2,o \in V$ be three vertices of degree $1$ in $O$.
    Let $e_1, e_2, f \in E$ be the edges containing $p_1, p_2$ and $o$, respectively.
    We say that $(O,p_1,p_2,o)$ is \kor{k}, if the following hold.
\begin{itemize}
    \item For any partition $\{\tilde{F}_k, \tilde{M}\}$ of the set $\{e_1,e_2 \}$ there exists a \mdec{k} $(F_k,M)$ of $O$ such that $\tilde{F}_k \subset F_k$ and $\tilde{M} \subset M$.
    \item For any \mdec{k} $(F_k,M)$ of $O$ in which $e_1, e_2 \in M$ we have $f \in M$.
    \item For any \mdec{k} $(F_k,M)$ of $O$ in which $e_1 \in F_k$ or $e_2 \in F_k$ we have $f \in F_k$.
\end{itemize}
\end{definition}

Such a gadget exists for all $k \geq 3$:

\begin{proposition}
\label{prop:ExistenceORgadget}
For all $k \geq 3$ there exists a \kor{k} $(O, p_1, p_2, o)$. 
Moreover, the statement still holds with the additional restriction of $O$ being a planar graph.
\end{proposition}

\begin{proof}
Construct the graph $O$ with vertices $p_1, p_2, o$ as in \Cref{fig:or}, where $(G_{k-2},v)$ is a \ffor{(k-2)}{k}.

As in any \mdec{k} $(F_k,M)$ of $O$ each vertex of degree $3$ has to be incident with exactly one edge in $M$, and as $F_k$ does not contain any paths on $k+1$ edges, it is straightforward to check that there are exactly $4$ possible \mdec{k}s of $O$; those are depicted in \Cref{fig:or_pf}.
This shows that $(O,p_1,p_2,o)$ is a \kor{k}.
Moreover, by taking $G_{k-2}$ to be a planar graph, which is possible by \Cref{prop:ellFforcer}, we get that $O$ is planar as well.
\end{proof}

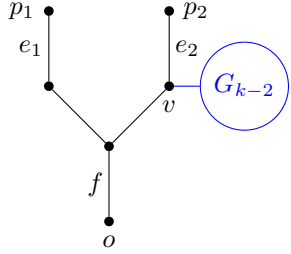
\begin{figure}
        \centering
	\begin{tikzpicture}[auto, vertex/.style={circle,draw=black!100,fill=black!100, thick,inner sep=0pt,minimum size=1mm}]
			\node (x) at (0,0) [vertex] {};
			\node (o) at (0,-1) [vertex,label=below:$o$] {};
			\node (v1) at (-0.8,0.8) [vertex] {};
			\node (v2) at (0.8,0.8) [vertex,label=below:$v$] {};
			\node (p1) at (-0.8,1.8) [vertex,label=left:$p_1$] {};
			\node (p2) at (0.8,1.8) [vertex,label=right:$p_2$] {};
                \node (G) at (1.8,0.8) [circle,draw=blue,inner ysep=0.3em] {\color{blue}$G_{k-2}$};
			
			\draw [-] (x) --node[inner sep=2pt,swap]{$f$} (o);
                \draw [-] (v1) --node[inner sep=2pt]{} (x);
                \draw [-] (v2) --node[inner sep=2pt]{} (x);
                \draw [-] (p1) --node[inner sep=2pt,swap]{$e_1$} (v1);
                \draw [-] (p2) --node[inner sep=2pt]{$e_2$} (v2);
                \draw [-,blue] (v2) --node[inner sep=2pt]{} (G);
		\end{tikzpicture}
	\caption{The graph $O$ with vertices $p_1, p_2, o$ of degree $1$, where $(G_{k-2},v)$ is an $((k-2)F,k)$-forcer.}
	\label{fig:or}
\end{figure}

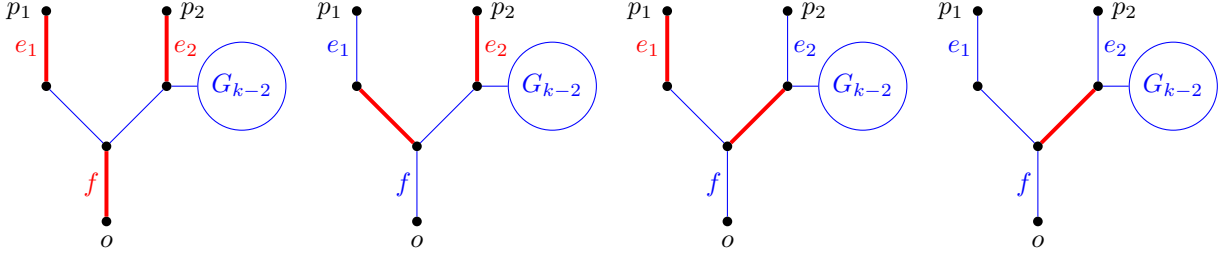
\begin{figure}
    \centering
    \begin{minipage}{.25\textwidth}
	\centering
	\begin{tikzpicture}[auto, vertex/.style={circle,draw=black!100,fill=black!100, thick,inner sep=0pt,minimum size=1mm}]
            \node (x) at (0,0) [vertex] {};
            \node (o) at (0,-1) [vertex,label=below:$o$] {};
            \node (v1) at (-0.8,0.8) [vertex] {};
            \node (v2) at (0.8,0.8) [vertex] {};
            \node (p1) at (-0.8,1.8) [vertex,label=left:$p_1$] {};
            \node (p2) at (0.8,1.8) [vertex,label=right:$p_2$] {};
            \node (G) at (1.8,0.8) [circle,draw=blue,inner ysep=0.3em] {\color{blue}$G_{k-2}$};
            
            \draw [-,red,line width=1.5pt] (x) --node[inner sep=2pt,swap]{$f$} (o);
            \draw [-,blue] (v1) --node[inner sep=2pt]{} (x);
            \draw [-,blue] (v2) --node[inner sep=2pt]{} (x);
            \draw [-,red,line width=1.5pt] (p1) --node[inner sep=2pt,swap]{$e_1$} (v1);
            \draw [-,red,line width=1.5pt] (p2) --node[inner sep=2pt]{$e_2$} (v2);
            \draw [-,blue] (v2) --node[inner sep=2pt]{} (G);
	\end{tikzpicture}
    \end{minipage}%
    \begin{minipage}{.25\textwidth}
	\centering
	\begin{tikzpicture}[auto, vertex/.style={circle,draw=black!100,fill=black!100, thick,inner sep=0pt,minimum size=1mm}]
            \node (x) at (0,0) [vertex] {};
            \node (o) at (0,-1) [vertex,label=below:$o$] {};
            \node (v1) at (-0.8,0.8) [vertex] {};
            \node (v2) at (0.8,0.8) [vertex] {};
            \node (p1) at (-0.8,1.8) [vertex,label=left:$p_1$] {};
            \node (p2) at (0.8,1.8) [vertex,label=right:$p_2$] {};
            \node (G) at (1.8,0.8) [circle,draw=blue,inner ysep=0.3em] {\color{blue}$G_{k-2}$};
            
            \draw [-,blue] (x) --node[inner sep=2pt,swap]{$f$} (o);
            \draw [-,red,line width=1.5pt] (v1) --node[inner sep=2pt]{} (x);
            \draw [-,blue] (v2) --node[inner sep=2pt]{} (x);
            \draw [-,blue] (p1) --node[inner sep=2pt,swap]{$e_1$} (v1);
            \draw [-,red,line width=1.5pt] (p2) --node[inner sep=2pt]{$e_2$} (v2);
            \draw [-,blue] (v2) --node[inner sep=2pt]{} (G);
	\end{tikzpicture}
    \end{minipage}%
    \begin{minipage}{.25\textwidth}
	\centering
	\begin{tikzpicture}[auto, vertex/.style={circle,draw=black!100,fill=black!100, thick,inner sep=0pt,minimum size=1mm}]
            \node (x) at (0,0) [vertex] {};
            \node (o) at (0,-1) [vertex,label=below:$o$] {};
            \node (v1) at (-0.8,0.8) [vertex] {};
            \node (v2) at (0.8,0.8) [vertex] {};
            \node (p1) at (-0.8,1.8) [vertex,label=left:$p_1$] {};
            \node (p2) at (0.8,1.8) [vertex,label=right:$p_2$] {};
            \node (G) at (1.8,0.8) [circle,draw=blue,inner ysep=0.3em] {\color{blue}$G_{k-2}$};
            
            \draw [-,blue] (x) --node[inner sep=2pt,swap]{$f$} (o);
            \draw [-,blue] (v1) --node[inner sep=2pt]{} (x);
            \draw [-,red,line width=1.5pt] (v2) --node[inner sep=2pt]{} (x);
            \draw [-,red,line width=1.5pt] (p1) --node[inner sep=2pt,swap]{$e_1$} (v1);
            \draw [-,blue] (p2) --node[inner sep=2pt]{$e_2$} (v2);
            \draw [-,blue] (v2) --node[inner sep=2pt]{} (G);
	\end{tikzpicture}
    \end{minipage}%
    \begin{minipage}{.25\textwidth}
	\centering
	\begin{tikzpicture}[auto, vertex/.style={circle,draw=black!100,fill=black!100, thick,inner sep=0pt,minimum size=1mm}]
            \node (x) at (0,0) [vertex] {};
            \node (o) at (0,-1) [vertex,label=below:$o$] {};
            \node (v1) at (-0.8,0.8) [vertex] {};
            \node (v2) at (0.8,0.8) [vertex] {};
            \node (p1) at (-0.8,1.8) [vertex,label=left:$p_1$] {};
            \node (p2) at (0.8,1.8) [vertex,label=right:$p_2$] {};
            \node (G) at (1.8,0.8) [circle,draw=blue,inner ysep=0.3em] {\color{blue}$G_{k-2}$};
            
            \draw [-,blue] (x) --node[inner sep=2pt,swap]{$f$} (o);
            \draw [-,blue] (v1) --node[inner sep=2pt]{} (x);
            \draw [-,red,line width=1.5pt] (v2) --node[inner sep=2pt]{} (x);
            \draw [-,blue] (p1) --node[inner sep=2pt,swap]{$e_1$} (v1);
            \draw [-,blue] (p2) --node[inner sep=2pt]{$e_2$} (v2);
            \draw [-,blue] (v2) --node[inner sep=2pt]{} (G);
	\end{tikzpicture}
    \end{minipage}
     \caption{The only four possible \mdec{k}s of $O$.
     The \ffor{(k-2)}{k} $G_{k-2}$ as well as the edges of the $k$-bounded linear forest are in blue, and the edges of the matching are in red and bold.}
    \label{fig:or_pf}
\end{figure}

\subsubsection{\rej{k}}

\begin{definition}
\label{def:kFrejector}
    Let $k \geq 3$.
    Let $R=(V,E)$ be a graph and $n', n \in V$ be vertices of degree $1$ in $R$.
    Let $e', e \in E$ be the edges containing $n', n$, respectively. 
    We say that $(R, n', n)$ is an \emph{\rej{k}}, if the following hold.
\begin{itemize}
    \item There is no \mdec{k} $(F_k,M)$ of $R$ such that $e', e \in F_k$.
    \item For any partition $\{\tilde{F}_k,\tilde{M} \}$ of the set $\{e',e \}$ such that at least one of $e',e$ is in $\tilde{M}$, there exists a \mdec{k} $(F_k,M)$ of $R$ for which $\tilde{F}_k \subset F_k$ and $\tilde{M} \subset M$, and moreover, neither $e'$ nor $e$ is contained in a path of length greater than $1$ in $F_k$.
\end{itemize}
\end{definition}

An important property of a \rej{k} is that it consists of a graph which admits a \mdec{k} if and only if at least one of its input edges is contained in the matching.
This will be used later in the final reduction to verify that the truth value $\true$ is not being assigned to both $x$ and $\bar{x}$.

\begin{proposition}
\label{prop:ExistenceRejGadget}
    For all $k \geq 3$ there exists an \rej{k} $(R, n', n)$.
    Moreover, the statement still holds with the additional restriction of $R$ being a planar graph.
\end{proposition}

\begin{proof}
Let $k\geq 3$.
We give an explicit construction of an \rej{k} $(R, n', n)$.
Let $\left\{(G_{k-2}^{(j)},v_j)\right\}_{j \in \khalf}$ be a collection of disjoint copies of the \ffor{(k-2)}{k} from the proof of \Cref{prop:ellFforcer}, and let $w_1, \ldots, w_{\khalf}, u_1, \ldots, u_{\khalf + 1}, n', n$ be $2\khalf +3$ distinct vertices, disjoint from all $G_{k-2}^{(j)}$'s.
Define the graph $R = (V,E)$ as follows (see \Cref{fig:ffrej}).
\begin{align*}
    V &\coloneqq \bigcup_{j=1}^{\khalf} V\left(G_{k-2}^{(j)} \right) \cup \{w_1, \ldots, w_{\khalf} \} \cup \{u_1, \ldots, u_{\khalf +1} \} \cup \{n', n \} \\
    E &\coloneqq \bigcup_{j=1}^{\khalf} \left(E\left(G_{k-2}^{(j)} \right)\cup \{v_j w_j, w_j u_j, w_j u_{j+1} \} \right) \cup \{e', e \}, 
\end{align*}

where $e' = n'u_1$ and $e = n u_{\khalf +1}$.
We claim that $(R,n',n)$ is an \rej{k}, by considering each item in \Cref{def:kFrejector}.

First, we show that there is no \mdec{k} $(F_k,M)$ of $R$ such that $e',e \in F_k$. 
Suppose otherwise and let $(F_k,M)$ be a \mdec{k} of $R$ such that $e',e \in F_k$.
We claim that there exists $1 \le j \le \khalf$ such that $v_j w_j \in F_k$.
Indeed, otherwise, for all $1 \le j \le \khalf$ we have $v_j w_j \in M$, implying that $u_jw_j, w_ju_{j+1} \in F_k$.
But then $(n',u_1,w_1,u_2,w_2, \ldots, w_{\khalf},u_{\khalf+1},n)$ would be a path of length $2\khalf +2 \ge k+1$ in $F_k$, which is a contradiction.

Let $j_0$ be the minimal $j$ for which $v_j w_j \in F_k$.
Then, for all $j < j_0$ we have $v_j w_j \in M$ and therefore $u_jw_j, w_ju_{j+1} \in F_k$.
Note that this implies $u_{j_0}w_{j_0} \in M$, as otherwise $(n',u_1,w_1,u_2, \ldots, w_{j_0}, v_{j_0})$ would be a path of length at least $3$ in $F_k$, connecting to a path of length $k-2$ in $G^{(j_0)}_{k-2} \cap F_k$ via $v_{j_0}$, forming a path of length greater than $k$ in $F_k$, a contradiction.
Since $u_{j_0}w_{j_0} \in M$, we have $w_{j_0}u_{j_0+1} \in F_k$.
Hence, since $v_{j_0}w_{j_0} \in F_k$ by the choice of $j_0$, we get that $(v_{j_0},w_{j_0},u_{j_0+1})$ is a path of length $2$ in $F_k$, connecting to a path of length $k-2$ in $G^{(j_0)}_{k-2} \cap F_k$ via $v_{j_0}$, forcing $u_{j_0+1}w_{j_0+1} \in M$.
Iterating this argument, we obtain that for all $j_0 \le j \le \khalf$ we have $u_jw_j \in M$ and $w_jv_j, w_ju_{j+1} \in F_k$.
But then $(n,u_{\khalf+1},w_{\khalf},v_{\khalf})$ is a path of length $3$ in $F_k$ connecting to a path of length $k-2$ in $G^{(\lfloor k/2 \rfloor)}_{k-2} \cap F_k$ via $v_{\khalf}$, forming a path of length greater than $k$ in $F_k$, a contradiction.

\vspace{10pt}

Secondly, given a partition $\{\tilde{F}_k, \tilde{M} \}$ of the set $\{e', e \}$ such that at least one of $e',e$ is in $\tilde{M}$, we give a \mdec{k} $(F_k,M)$ of $R$ satisfying $\tilde{F}_k \subset F_k$ and $\tilde{M} \subset M$.
For each $1 \le j \le \khalf$ we fix a \mdec{k} $(F_k^{(j)}, M^{(j)})$ of $G^{(j)}_{k-2}$.
By \Cref{prop:ellFforcer} and \Cref{def:ellFforcer} we know that $v_j$ is the end-vertex of a $(k-2)$-path in $F_{k}^{(j)}$ for every $1 \le j \le \khalf$.
We now consider all three possibilities for $\{\tilde{F}_k, \tilde{M} \}$ and we decompose the remaining edges to get $(F_k,M)$ as follows (see \Cref{fig:FFrejAssignments}).
\begin{enumerate}
    \item[(a)] $\tilde{F}_k = \{e'\}$; $\tilde{M} = \{e \}$.
    For all $j$, we let $u_jw_j \in M$ and $v_jw_j, w_ju_{j+1} \in F_k$.
    \item[(b)] $\tilde{F}_k = \{e\}$; $\tilde{M} = \{e' \}$.
    For all $j$, we let $w_ju_{j+1} \in M$ and $v_jw_j, u_jw_j \in F_k$.
    \item[(c)] $\tilde{F}_k = \emptyset$; $\tilde{M} = \{e',e \}$.
    For all $j$ we let $v_jw_j \in M$ and $u_jw_j, w_ju_{j+1} \in F_k$.
\end{enumerate}

It is straightforward to check that in all three cases $(F_k,M)$ is a valid \mdec{k} of $R$, noting that in the third case the path $u_1w_1u_2 \ldots w_{\khalf}u_{\khalf +1} \in F_k$ has length $2 \khalf \leq k$.
Note also that in all three cases neither $e'$ nor $e$ is contained in a path of length greater than $1$ in $F_k$.
Hence, $(R, n', n)$ is an \rej{k}.
Moreover, by taking all $G_{k-2}^{(j)}$'s to be planar graphs, which is possible by \Cref{prop:ellFforcer}, we get a graph $R$ which is planar as well.
\end{proof}

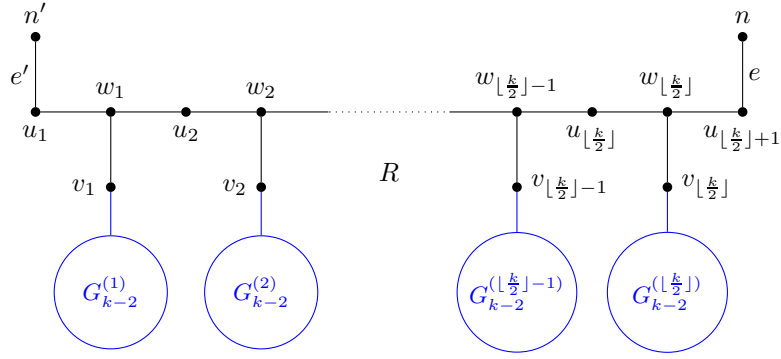
\begin{figure}
    \centering
	\begin{tikzpicture}[auto, vertex/.style={circle,draw=black!100,fill=black!100, thick,inner sep=0pt,minimum size=1mm}, novertex/.style={circle,draw=black!100,fill=black!100,inner sep=0pt,minimum size=0mm}]
			\node (n') at (-4.7,1) [vertex,label=$n'$] {};
                \node (n) at (4.7,1) [vertex,label=$n$] {};
   
                \node (x) at (-0.8,0) [novertex] {};
                \node (y) at (0.8,0) [novertex] {};
                \node (w2) at (-1.7,0) [vertex,label=$w_2$] {};
			\node (u2) at (-2.7,0) [vertex,label=below:$u_2$] {};
			\node (w1) at (-3.7,0) [vertex,label=$w_1$] {};
			\node (u1) at (-4.7,0) [vertex,label=below:$u_1$] {};
			\node (w3) at (1.7,0) [vertex,label=$w_{\khalf-1}$] {};
			\node (u4) at (2.7,0) [vertex,label=below:$u_{\khalf}$] {};
			\node (w4) at (3.7,0) [vertex,label=$w_{\khalf}$] {};
			\node (u5) at (4.7,0) [vertex,label=below:$u_{\khalf +1}$] {};

                \node (v1) at (-3.7,-1) [vertex,label=left:$v_1$] {};
                \node (v2) at (-1.7,-1) [vertex,label=left:$v_2$] {};
                \node (v3) at (1.7,-1) [vertex,label=right:$v_{\khalf -1}$] {};
                \node (v4) at (3.7,-1) [vertex,label=right:$v_{\khalf}$] {};
                
                \node (G1) at (-3.7,-2.4) [circle,draw=blue,inner ysep=1em] {\color{blue}\small{$G^{(1)}_{k-2}$}};
                \node (G2) at (-1.7,-2.4) [circle,draw=blue,inner ysep=1em] {\color{blue}\small{$G^{(2)}_{k-2}$}};
                \node (G3) at (1.7,-2.4) [circle,draw=blue,inner ysep=0em] {\color{blue}\small{$G^{(\khalf-1)}_{k-2}$}};
                \node (G4) at (3.7,-2.4) [circle,draw=blue,inner ysep=0.8em] {\color{blue}\small{$G^{(\khalf)}_{k-2}$}};
                \node (R) at (0,-0.8) [] {$R$};
			
			\draw [-,blue] (v1) --node[inner sep=2pt]{} (G1);
			\draw [-] (v1) --node[inner sep=2pt]{} (w1);
                \draw [-] (w1) --node[inner sep=2pt]{} (u1);
			\draw [-] (w1) --node[inner sep=2pt]{} (u2);
                
                \draw [-,blue] (v2) --node[inner sep=2pt]{} (G2);
			\draw [-] (v2) --node[inner sep=2pt]{} (w2);
                \draw [-] (w2) --node[inner sep=2pt]{} (u2);
			\draw [-] (w2) --node[inner sep=2pt]{} (x);
   
			\draw [-,dotted] (x) --node[inner sep=2pt,swap]{} (y);
   
                \draw [-,blue] (v3) --node[inner sep=2pt]{} (G3);
			\draw [-] (v3) --node[inner sep=2pt]{} (w3);
                \draw [-] (w3) --node[inner sep=2pt]{} (y);
			\draw [-] (w3) --node[inner sep=2pt]{} (u4);
                
                \draw [-,blue] (v4) --node[inner sep=2pt]{} (G4);
			\draw [-] (v4) --node[inner sep=2pt]{} (w4);
                \draw [-] (w4) --node[inner sep=2pt]{} (u4);
			\draw [-] (w4) --node[inner sep=2pt]{} (u5);

                \draw [-] (n') --node[inner sep=2pt,swap]{$e'$} (u1);
                \draw [-] (n) --node[inner sep=2pt]{$e$} (u5);
	\end{tikzpicture}
    \caption{An \rej{k} $(R, n', n)$ consisting of a graph $R$ with vertices $n', n$ of degree $1$.}
    \label{fig:ffrej}
\end{figure}

\begin{figure}
\centering
\begin{subfloat}[$\tilde{F}_k = \{e\}$; $\tilde{M} = \{e' \}$] {
\centering
    \begin{tikzpicture}[auto, vertex/.style={circle,draw=black!100,fill=black!100, thick,inner sep=0pt,minimum size=0.7mm}, novertex/.style={circle,draw=black!100,fill=black!100,inner sep=0pt,minimum size=0mm}]
			\node (n') at (-2.35,0.7) [vertex,label=$n'$] {};
                \node (n) at (2.35,0.7) [vertex,label=$n$] {};
   
                \node (x) at (-0.4,0) [novertex] {};
                \node (y) at (0.4,0) [novertex] {};
                \node (w2) at (-0.85,0) [vertex] {};
			\node (u2) at (-1.35,0) [vertex] {};
			\node (w1) at (-1.85,0) [vertex] {};
			\node (u1) at (-2.35,0) [vertex] {};
			\node (w3) at (0.85,0) [vertex] {};
			\node (u4) at (1.35,0) [vertex] {};
			\node (w4) at (1.85,0) [vertex] {};
			\node (u5) at (2.35,0) [vertex] {};

                \node (v1) at (-1.85,-0.5) [vertex] {};
                \node (v2) at (-0.85,-0.5) [vertex] {};
                \node (v3) at (0.85,-0.5) [vertex] {};
                \node (v4) at (1.85,-0.5) [vertex] {};
                
                \node (G1) at (-1.85,-1.2) [circle,draw=blue,inner ysep=1em] {};
                \node (G2) at (-0.85,-1.2) [circle,draw=blue,inner ysep=1em] {};
                \node (G3) at (0.85,-1.2) [circle,draw=blue,inner ysep=1em] {};
                \node (G4) at (1.85,-1.2) [circle,draw=blue,inner ysep=1em] {};
			
			\draw [-,blue] (v1) --node[inner sep=2pt]{} (G1);
			\draw [-,blue] (v1) --node[inner sep=2pt]{} (w1);
                \draw [-,blue] (w1) --node[inner sep=2pt]{} (u1);
			\draw [-,red,line width=1.5pt] (w1) --node[inner sep=2pt]{} (u2);
                
                \draw [-,blue] (v2) --node[inner sep=2pt]{} (G2);
			\draw [-,blue] (v2) --node[inner sep=2pt]{} (w2);
                \draw [-,blue] (w2) --node[inner sep=2pt]{} (u2);
			\draw [-,red,line width=1.5pt] (w2) --node[inner sep=2pt]{} (x);
   
			\draw [-,dotted] (x) --node[inner sep=2pt,swap]{} (y);
   
                \draw [-,blue] (v3) --node[inner sep=2pt]{} (G3);
			\draw [-,blue] (v3) --node[inner sep=2pt]{} (w3);
                \draw [-,blue] (w3) --node[inner sep=2pt]{} (y);
			\draw [-,red,line width=1.5pt] (w3) --node[inner sep=2pt]{} (u4);
                
                \draw [-,blue] (v4) --node[inner sep=2pt]{} (G4);
			\draw [-,blue] (v4) --node[inner sep=2pt]{} (w4);
                \draw [-,blue] (w4) --node[inner sep=2pt]{} (u4);
			\draw [-,red,line width=1.5pt] (w4) --node[inner sep=2pt]{} (u5);

                \draw [-,red,line width=1.5pt] (n') --node[inner sep=2pt,swap]{$e'$} (u1);
                \draw [-,blue] (n) --node[inner sep=2pt]{$e$} (u5);
	\end{tikzpicture}
}
\end{subfloat}
\hspace{10pt}       
\begin{subfloat}[$\tilde{F}_k = \{e'\}$; $\tilde{M} = \{e \}$] {
\centering
    \begin{tikzpicture}[auto, vertex/.style={circle,draw=black!100,fill=black!100, thick,inner sep=0pt,minimum size=0.7mm}, novertex/.style={circle,draw=black!100,fill=black!100,inner sep=0pt,minimum size=0mm}]
			\node (n') at (-2.35,0.7) [vertex,label=$n'$] {};
                \node (n) at (2.35,0.7) [vertex,label=$n$] {};
   
                \node (x) at (-0.4,0) [novertex] {};
                \node (y) at (0.4,0) [novertex] {};
                \node (w2) at (-0.85,0) [vertex] {};
			\node (u2) at (-1.35,0) [vertex] {};
			\node (w1) at (-1.85,0) [vertex] {};
			\node (u1) at (-2.35,0) [vertex] {};
			\node (w3) at (0.85,0) [vertex] {};
			\node (u4) at (1.35,0) [vertex] {};
			\node (w4) at (1.85,0) [vertex] {};
			\node (u5) at (2.35,0) [vertex] {};

                \node (v1) at (-1.85,-0.5) [vertex] {};
                \node (v2) at (-0.85,-0.5) [vertex] {};
                \node (v3) at (0.85,-0.5) [vertex] {};
                \node (v4) at (1.85,-0.5) [vertex] {};
                
                \node (G1) at (-1.85,-1.2) [circle,draw=blue,inner ysep=1em] {};
                \node (G2) at (-0.85,-1.2) [circle,draw=blue,inner ysep=1em] {};
                \node (G3) at (0.85,-1.2) [circle,draw=blue,inner ysep=1em] {};
                \node (G4) at (1.85,-1.2) [circle,draw=blue,inner ysep=1em] {};

			\draw [-,blue] (v1) --node[inner sep=2pt]{} (G1);
			\draw [-,blue] (v1) --node[inner sep=2pt]{} (w1);
                \draw [-,red,line width=1.5pt] (w1) --node[inner sep=2pt]{} (u1);
			\draw [-,blue] (w1) --node[inner sep=2pt]{} (u2);
                
                \draw [-,blue] (v2) --node[inner sep=2pt]{} (G2);
			\draw [-,blue] (v2) --node[inner sep=2pt]{} (w2);
                \draw [-,red,line width=1.5pt] (w2) --node[inner sep=2pt]{} (u2);
			\draw [-,blue] (w2) --node[inner sep=2pt]{} (x);
   
			\draw [-,dotted] (x) --node[inner sep=2pt,swap]{} (y);
   
                \draw [-,blue] (v3) --node[inner sep=2pt]{} (G3);
			\draw [-,blue] (v3) --node[inner sep=2pt]{} (w3);
                \draw [-,red,line width=1.5pt] (w3) --node[inner sep=2pt]{} (y);
			\draw [-,blue] (w3) --node[inner sep=2pt]{} (u4);
                
                \draw [-,blue] (v4) --node[inner sep=2pt]{} (G4);
			\draw [-,blue] (v4) --node[inner sep=2pt]{} (w4);
                \draw [-,red,line width=1.5pt] (w4) --node[inner sep=2pt]{} (u4);
			\draw [-,blue] (w4) --node[inner sep=2pt]{} (u5);

                \draw [-,blue] (n') --node[inner sep=2pt,swap]{$e'$} (u1);
                \draw [-,red,line width=1.5pt] (n) --node[inner sep=2pt]{$e$} (u5);
	\end{tikzpicture}
}       
\end{subfloat}
\vspace{20pt}%
\begin{subfloat}[$\tilde{F}_k = \emptyset$; $\tilde{M} = \{e',e \}$] {
\centering
    \begin{tikzpicture}[auto, vertex/.style={circle,draw=black!100,fill=black!100, thick,inner sep=0pt,minimum size=0.7mm}, novertex/.style={circle,draw=black!100,fill=black!100,inner sep=0pt,minimum size=0mm}]
			\node (n') at (-2.35,0.7) [vertex,label=$n'$] {};
                \node (n) at (2.35,0.7) [vertex,label=$n$] {};
   
                \node (x) at (-0.4,0) [novertex] {};
                \node (y) at (0.4,0) [novertex] {};
                \node (w2) at (-0.85,0) [vertex] {};
			\node (u2) at (-1.35,0) [vertex] {};
			\node (w1) at (-1.85,0) [vertex] {};
			\node (u1) at (-2.35,0) [vertex] {};
			\node (w3) at (0.85,0) [vertex] {};
			\node (u4) at (1.35,0) [vertex] {};
			\node (w4) at (1.85,0) [vertex] {};
			\node (u5) at (2.35,0) [vertex] {};

                \node (v1) at (-1.85,-0.5) [vertex] {};
                \node (v2) at (-0.85,-0.5) [vertex] {};
                \node (v3) at (0.85,-0.5) [vertex] {};
                \node (v4) at (1.85,-0.5) [vertex] {};
                
                \node (G1) at (-1.85,-1.2) [circle,draw=blue,inner ysep=1em] {};
                \node (G2) at (-0.85,-1.2) [circle,draw=blue,inner ysep=1em] {};
                \node (G3) at (0.85,-1.2) [circle,draw=blue,inner ysep=1em] {};
                \node (G4) at (1.85,-1.2) [circle,draw=blue,inner ysep=1em] {};
                \node (space) at (0,1.7) [novertex] {};
			
			\draw [-,blue] (v1) --node[inner sep=2pt]{} (G1);
			\draw [-,red,line width=1.5pt] (v1) --node[inner sep=2pt]{} (w1);
                \draw [-,blue] (w1) --node[inner sep=2pt]{} (u1);
			\draw [-,blue] (w1) --node[inner sep=2pt]{} (u2);
                
                \draw [-,blue] (v2) --node[inner sep=2pt]{} (G2);
			\draw [-,red,line width=1.5pt] (v2) --node[inner sep=2pt]{} (w2);
                \draw [-,blue] (w2) --node[inner sep=2pt]{} (u2);
			\draw [-,blue] (w2) --node[inner sep=2pt]{} (x);
   
			\draw [-,dotted] (x) --node[inner sep=2pt,swap]{} (y);
   
                \draw [-,blue] (v3) --node[inner sep=2pt]{} (G3);
			\draw [-,red,line width=1.5pt] (v3) --node[inner sep=2pt]{} (w3);
                \draw [-,blue] (w3) --node[inner sep=2pt]{} (y);
			\draw [-,blue] (w3) --node[inner sep=2pt]{} (u4);
                
                \draw [-,blue] (v4) --node[inner sep=2pt]{} (G4);
			\draw [-,red,line width=1.5pt] (v4) --node[inner sep=2pt]{} (w4);
                \draw [-,blue] (w4) --node[inner sep=2pt]{} (u4);
			\draw [-,blue] (w4) --node[inner sep=2pt]{} (u5);

                \draw [-,red,line width=1.5pt] (n') --node[inner sep=2pt,swap]{$e'$} (u1);
                \draw [-,red,line width=1.5pt] (n) --node[inner sep=2pt]{$e$} (u5);
	\end{tikzpicture}
}       
\end{subfloat}
\caption{Possible \mdec{k}s $(F_k,M)$ of $R$ such that $\tilde{F}_k \subset F_k$ and $\tilde{M} \subset M$. Copies of a \ffor{(k-2)}{k} as well as $F_k$ are coloured blue and $M$ is coloured red and in bold.}
\label{fig:FFrejAssignments}
\end{figure}
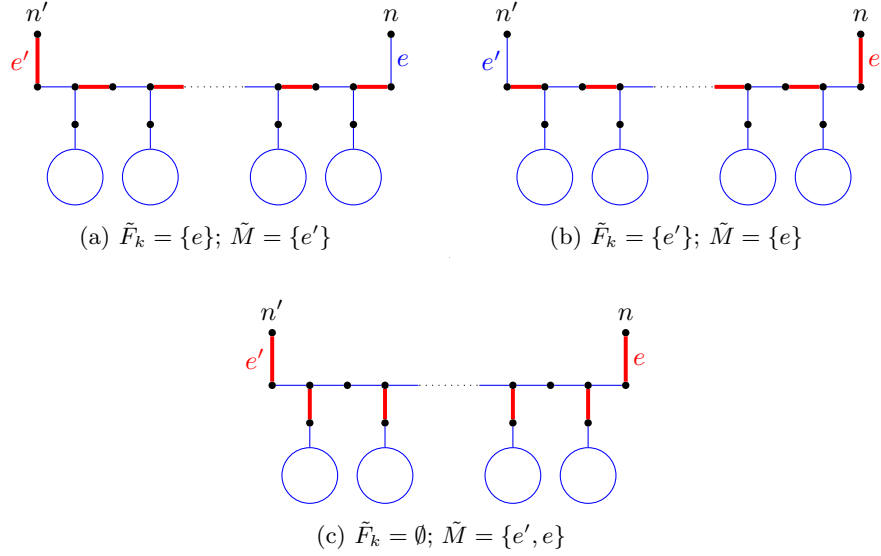

\subsection{Constructing a variable gadget}
\label{sec:VarGadget}

We construct the variable gadget that we shall later use in the final reduction.

\begin{construction}
\label{constr:VariableGadgetReduction}
Let $k \ge 3$.
We construct a graph $H$ as follows.
Let $(R,n',n)$ be an \rej{k} such that $R$ is a planar graph, let $e'$ and $e$ be the edges in $R$ containing $n'$ and $n$, respectively, and let $o$ be the vertex for which $e' = n'o$.
Let $(O,p_1,p_2,o)$ be a \kor{k}, such that $e'$ is the edge in $O$ containing $o$, and denote by $e_1$ and $e_2$ the edges in $O$ containing $p_1$ and $p_2$, respectively.
Note that we have $V(O) \cap V(R) = \{n',o\}$ and $E(O) \cap E(R) = \{e' \}$.
Let $H \coloneqq O \cup R$ (see \Cref{fig:gadget}).
We say that $(H,p_1,p_2,r)$ is a \emph{variable gadget}, $e_1, e_2$ and $p_1, p_2$ are its \emph{positive input edges} and \emph{positive input vertices}, respectively, and $e, r$ are its \emph{negative input edge} and its \emph{negative input vertex}, respectively.
\end{construction}

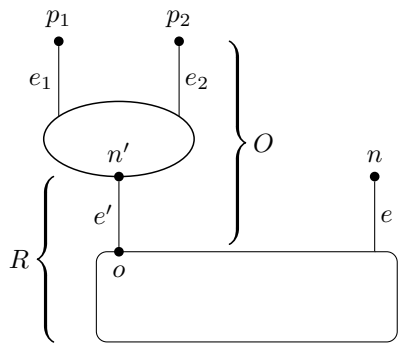
\begin{figure}
    \centering
	\begin{tikzpicture}[auto, vertex/.style={circle,draw=black!100,fill=black!100, thick,inner sep=0pt,minimum size=1mm}, novertex/.style={circle,draw=black!100,fill=black!100,inner sep=0pt,minimum size=0mm},decoration={calligraphic brace, mirror, amplitude=6pt,raise=2pt}]

            \draw [line width=0.2mm] (-1.2,2) ellipse (1cm and 0.5cm);
            \node (n') at (-1.2,1.5) [vertex,label=above:$n'$] {};
            \node (o) at (-1.2,0.5) [vertex,label=below:$o$] {};
            \node (x1) at (-2,2.3) [novertex] {};
            \node (x2) at (-0.4,2.3) [novertex] {};
            \node (p1) at (-2,3.3) [vertex,label=above:$p_1$] {};
            \node (p2) at (-0.4,3.3) [vertex,label=above:$p_2$] {};
            
            \draw [-] (n') --node[inner sep=2pt,swap]{$e'$} (o);
            \draw [-] (p1) --node[inner sep=2pt,swap]{$e_1$} (x1);
            \draw [-] (p2) --node[inner sep=2pt]{$e_2$} (x2);

            \draw [decorate,very thick]
            (0.2,0.6) -- (0.2,3.3) node [black,midway,xshift=0.8cm] {
            $O$};

            \draw[rounded corners] (-1.5,-0.7) rectangle (2.5,0.5) {};
            \node (o1) at (2.2,0.5) [novertex] {};
            \node (n) at (2.2,1.5) [vertex,label=$n$] {};

            \draw [-] (n) --node[inner sep=2pt]{$e$} (o1);

            \draw [decorate,very thick]
            (-2,1.5) -- (-2,-0.7) node [black,midway,xshift=-0.8cm] {$R$};
            
	\end{tikzpicture}
     \caption{A variable gadget $(H,p_1,p_2,n)$, consisting of a \kor{k} $(O,p_1,p_2,o)$ and an \rej{k} $(R,n',n)$, such that $O$ and $R$ share the edge $e'=n'o$.}
    \label{fig:gadget}
\end{figure}

The following lemma captures the most important property of a variable gadget.

\begin{lemma}
\label{lem:FillingVariableGadget}
Let $k \ge 3$, and let $(H,p_1,p_2,n)$ be a copy of the variable gadget described in \Cref{constr:VariableGadgetReduction}, with $e_1, e_2$ as its positive input edges and $e$ as its negative input edge.
Let $\{\tilde{F}_k, \tilde{M} \}$ be a partition of the set $\{e_1, e_2, e \}$, such that neither $\{e_1, e \}$ nor $\{e_2, e \}$ is contained in $\tilde{F}_k$.
Then there exists a \mdec{k} $(F_k,M)$ of $H$ such that $\tilde{F}_k \subset F_k$, $\tilde{M} \subset M$.

Conversely, given an \mdec{k} $(F_k,M)$ of $H$, we have that neither of $\{e_1, e \}, \{e_2, e \}$ is contained in $F_k$.
\end{lemma}

\begin{proof}
    This lemma is a corollary of \Cref{prop:ExistenceORgadget,prop:ExistenceRejGadget}.
    Let $(\tilde{F}_k,\tilde{M})$ be a partition of the set $\{e_1, e_2, e\}$ satisfying the assumption in the statement, and let $(\tilde{F}_k^{O}, \tilde{M}^O)$ the induced partition on $\{e_1, e_2\}$.
    By \Cref{prop:ExistenceORgadget} there exists a \mdec{k} $(F_k^O, M^O)$ of $O$ for which $\tilde{F}_k^O \subset F_k^O$ and $\tilde{M}^O \subset M^O$.
    Let $\{\tilde{F}_k^R, \tilde{M}^R \}$ the partition of the set $\{e', e\}$ induced by $\{\tilde{F}_k\cup F_k^O, \tilde{M}\cup M^O \}$.
    By \Cref{def:kFrejector} there exists a \mdec{k} $(F_k^R, M^R)$ for which $\tilde{F}_k^R \subset F_k^R$ and $\tilde{M}^R \subset M^R$ if and only if at least one of $e',e$ is in $\tilde{M}^R$.
    However, since $(O,p_1,p_2,o)$ is a \kor{k}, this is equivalent to having neither $\{e_1, e \}$ nor $\{e_2, e \}$ contained in $\tilde{F}_k^R$, which is guaranteed by our assumptions.
    Let $(F_k^R, M^R)$ be such a \mdec{k}, and set $(F_k, M) \coloneqq (F_k^O\cup F_k^R, M^O\cup M^R)$.
    We then have $\tilde{F}_k \subset \tilde{F}_k^O\cup \tilde{F}_k^R \subset F_k^O\cup F_k^R = F_k$ and $\tilde{M} \subset \tilde{M}^O\cup \tilde{M}^R \subset M^O\cup M^R = M$, as required.

    Conversely, let $(F_k,M)$ be a \mdec{k} of $H$.
    If $e \in F_k$ then since $(R,n',n)$ is an \rej{k} we must have $e' \in M$.
    But then, as $(O,p_1,p_2,o)$ is a \kor{k}, we get that $e_1,e_2 \in M$, proving the statement.
\end{proof}

\subsection{The reduction}
\label{sec:Reduction}

We now finally describe the graph that we use to establish the reduction.

\begin{construction}
\label{constr:GraphReductionNP}
    Let $(X,\cC)$ be an instance of Planar $(\le 3,3)$–SAT and $k \ge 3$.
    We construct the planar graph $G_{(X,\cC)}$ as follows.
    For every clause $C \in \cC$ let $v_C$ be a vertex.
    For every variable $x \in X$ such that $x$ appears in $C_1, C_2$ and $\bar{x}$ appears in $C_3$, for some $C_1, C_2, C_3 \in \cC$, let $p^x_{C_1}$ and $p^x_{C_2}$ be positive input vertices, and $n^x_{C_3}$ be a negative input vertex.
    Then, let $(H^x,p^x_{C_1},p^x_{C_2},n^x_{C_3})$ be a copy of the variable gadget constructed in \ref{constr:VariableGadgetReduction}.
    Furthermore, for every clause $C \in \cC$ with $|C|=2$, attach to $v_C$ a \ffor{1}{k} $(G_C,v_C)$.
    Then $G_{(X,\cC)}$ is defined to be the union of all defined above, with the edges connecting the input vertices of each variable gadget to the relevant clause vertices, namely,
    \[\{v_{C_1} p^x_{C_1}, v_{C_2} p^x_{C_2}, v_{C_3} n^x_{C_3} : x \in X, \, C_1, C_2, C_3 \in \cC \text{ s.t. } x \in C_1,C_2;\, \bar{x}\in C_3 \}.\]
    See \Cref{fig:GXC} for an illustration.
\end{construction}

\begin{figure}
        \centering	\begin{tikzpicture}[auto, vertex/.style={circle,draw=black!100,fill=black!100, thick,inner sep=0pt,minimum size=1mm}, novertex/.style={circle,draw=black!100,fill=black!100,inner sep=0pt,minimum size=0mm}]]

            \draw [line width=0.25mm] (-6.25,0) ellipse (1cm and 0.5cm);
		\node [label={[shift={(0,-0.35)}]$x_1$}] at (-6.25,0) {};
            \node (p11) at (-7,1) [vertex] {};
            \node (p12) at (-6.5,1) [vertex] {};
            \node (n1) at (-5.5,1) [vertex] {};
            \node (v11) at (-7,0.33072) [novertex] {};
            \node (v12) at (-6.5,0.48412) [novertex] {};
            \node (u1) at (-5.5,0.33072) [novertex] {};

            \draw [-,line width=1.5pt] (p11) --node[inner sep=2pt]{} (v11);
            \draw [-,line width=1.5pt] (p12) --node[inner sep=2pt]{} (v12);
            \draw [-,dashed] (n1) --node[inner sep=2pt]{} (u1);

            \draw [line width=0.25mm] (-3.75,0) ellipse (1cm and 0.5cm);
		\node [label={[shift={(0,-0.35)}]$x_2$}] at (-3.75,0) {};
            \node (p21) at (-4.5,1) [vertex] {};
            \node (p22) at (-4,1) [vertex] {};
            \node (n2) at (-3,1) [vertex] {};
            \node (v21) at (-4.5,0.33072) [novertex] {};
            \node (v22) at (-4,0.48412) [novertex] {};
            \node (u2) at (-3,0.33072) [novertex] {};

            \draw [-,line width=1.5pt] (p21) --node[inner sep=2pt]{} (v21);
            \draw [-,line width=1.5pt] (p22) --node[inner sep=2pt]{} (v22);
            \draw [-,dashed] (n2) --node[inner sep=2pt]{} (u2);

            \draw [line width=0.25mm] (-1.25,0) ellipse (1cm and 0.5cm);
		\node [label={[shift={(0,-0.35)}]$x_3$}] at (-1.25,0) {};
            \node (p31) at (-2,1) [vertex] {};
            \node (p32) at (-1.5,1) [vertex] {};
            \node (n3) at (-0.5,1) [vertex] {};
            \node (v31) at (-2,0.33072) [novertex] {};
            \node (v32) at (-1.5,0.48412) [novertex] {};
            \node (u3) at (-0.5,0.33072) [novertex] {};

            \draw [-,line width=1.5pt] (p31) --node[inner sep=2pt]{} (v31);
            \draw [-,line width=1.5pt] (p32) --node[inner sep=2pt]{} (v32);
            \draw [-,dashed] (n3) --node[inner sep=2pt]{} (u3);

            \draw [line width=0.25mm] (1.25,0) ellipse (1cm and 0.5cm);
		\node [label={[shift={(0,-0.35)}]$x_4$}] at (1.25,0) {};
            \node (p41) at (0.5,1) [vertex] {};
            \node (p42) at (1,1) [vertex] {};
            \node (n4) at (2,1) [vertex] {};
            \node (v41) at (0.5,0.33072) [novertex] {};
            \node (v42) at (1,0.48412) [novertex] {};
            \node (u4) at (2,0.33072) [novertex] {};

            \draw [-,line width=1.5pt] (p41) --node[inner sep=2pt]{} (v41);
            \draw [-,line width=1.5pt] (p42) --node[inner sep=2pt]{} (v42);
            \draw [-,dashed] (n4) --node[inner sep=2pt]{} (u4);

            \draw [line width=0.25mm] (3.75,0) ellipse (1cm and 0.5cm);
		\node [label={[shift={(0,-0.35)}]$x_5$}] at (3.75,0) {};
            \node (p51) at (3,1) [vertex] {};
            \node (p52) at (3.5,1) [vertex] {};
            \node (n5) at (4.5,1) [vertex] {};
            \node (v51) at (3,0.33072) [novertex] {};
            \node (v52) at (3.5,0.48412) [novertex] {};
            \node (u5) at (4.5,0.33072) [novertex] {};

            \draw [-,line width=1.5pt] (p51) --node[inner sep=2pt]{} (v51);
            \draw [-,line width=1.5pt] (p52) --node[inner sep=2pt]{} (v52);
            \draw [-,dashed] (n5) --node[inner sep=2pt]{} (u5);

            \draw [line width=0.25mm] (6.25,0) ellipse (1cm and 0.5cm);
		\node [label={[shift={(0,-0.35)}]$x_6$}] at (6.25,0) {};
            \node (p61) at (5.5,1) [vertex] {};
            \node (p62) at (6,1) [vertex] {};
            \node (n6) at (7,1) [vertex] {};
            \node (v61) at (5.5,0.33072) [novertex] {};
            \node (v62) at (6,0.48412) [novertex] {};
            \node (u6) at (7,0.33072) [novertex] {};

            \draw [-,line width=1.5pt] (p61) --node[inner sep=2pt]{} (v61);
            \draw [-,line width=1.5pt] (p62) --node[inner sep=2pt]{} (v62);
            \draw [-,dashed] (n6) --node[inner sep=2pt]{} (u6);

            \node (c1) at (-7.5,3.75) [vertex,label=right:$C_1$] {};
            \draw [-] (c1) --node[inner sep=2pt]{} (n1);
            \draw [-] (c1) --node[inner sep=2pt]{} (p21);
            \node (F1) at (-7.5,4.5) [circle,draw=blue,inner ysep=0.8em] {};
            \draw [-,blue] (c1) --node[inner sep=2pt]{} (F1);

            \node (c2) at (-5,3.75) [vertex,label=right:$C_2$] {};
            \draw [-] (c2) --node[inner sep=2pt]{} (p11);
            \draw [-] (c2) --node[inner sep=2pt]{} (p32);
            \node (F2) at (-5,4.5) [circle,draw=blue,inner ysep=0.8em] {};
            \draw [-,blue] (c2) --node[inner sep=2pt]{} (F2);
            
            \node (c3) at (-2.5,3.75) [vertex,label=right:$C_3$] {};
            \draw [-] (c3) --node[inner sep=2pt]{} (p12);
            \draw [-] (c3) --node[inner sep=2pt]{} (n2);
            \draw [-] (c3) --node[inner sep=2pt]{} (p31);
            
            \node (c4) at (0,3.75) [vertex,label=right:$C_4$] {};
            \draw [-] (c4) --node[inner sep=2pt]{} (p22);
            \draw [-] (c4) --node[inner sep=2pt]{} (n4);
            \draw [-] (c4) --node[inner sep=2pt]{} (p51);
            
            \node (c5) at (2.5,3.75) [vertex,label=right:$C_5$] {};
            \draw [-] (c5) --node[inner sep=2pt]{} (n3);
            \draw [-] (c5) --node[inner sep=2pt]{} (p41);
            \draw [-] (c5) --node[inner sep=2pt]{} (p62);
            
            \node (c6) at (5,3.75) [vertex,label=right:$C_6$] {};
            \draw [-] (c6) --node[inner sep=2pt]{} (p42);
            \draw [-] (c6) --node[inner sep=2pt]{} (n5);
            \draw [-] (c6) --node[inner sep=2pt]{} (p61);
            
            \node (c7) at (7.5,3.75) [vertex,label=right:$C_7$] {};
            \draw [-] (c7) --node[inner sep=2pt]{} (p52);
            \draw [-] (c7) --node[inner sep=2pt]{} (n6);
            \node (F7) at (7.5,4.5) [circle,draw=blue,inner ysep=0.8em] {};
            \draw [-,blue] (c7) --node[inner sep=2pt]{} (F7);
	\end{tikzpicture}
	\caption{An example of the graph $G_{(X,\cC)}$ for $X = \{x_i \}_{i=1}^6$ and $\cC = \{C_i \}_{i=1}^7$ where $C_1 = \bar{x}_1 \vee x_2$; $C_2 = x_1 \vee x_3$; $C_3 = x_1 \vee \bar{x}_2 \vee x_3$; $C_4 = x_2 \vee \bar{x}_4 \vee x_5$; $C_5 = \bar{x}_3 \vee x_4 \vee x_6$; $C_6 = x_4 \vee \bar{x}_5 \vee x_6$; $C_7 = x_5 \vee \bar{x}_6$. The graph consist of $6$ variable gadgets, $7$ clause vertices, and $3$ \ffor{1}{k} for the $3$ clauses of size $2$. One can check that the graph is planar.}
    \label{fig:GXC}
\end{figure}
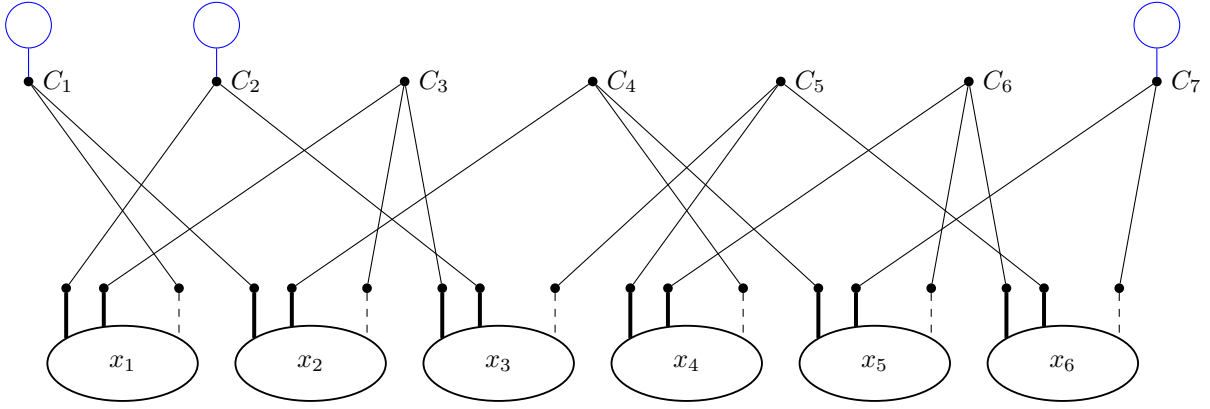

\begin{lemma}
\label{lem:3-SATtoMatchingForest}
     Let $k \ge 3$.
     Let $(X,\cC)$ be an instance of Planar $(\le 3,3)$–SAT, and let the graph $G_{(X,\cC)}$ be as in \Cref{constr:GraphReductionNP}.
     If $(X,\cC)$ is a yes-instance of Planar $(\le 3,3)$–SAT, then $G_{(X,\cC)}$ is a yes-instance of Planar \blfd{k}{1}.
\end{lemma}

\begin{proof}
    Let $\phi : X \rightarrow \left\{\true, \false \right\}$ be a truth assignment such that every clause of $\cC$ contains at least one true literal.
    We construct a valid \mdec{k} $(F_k,M)$ of the graph $G_{(X,\cC)}$ as follows.

    Let $x \in X$ be a variable and let $C_1, C_2, C_3 \in \cC$ be clauses such that $x$ appears in $C_1, C_2$ and $\bar{x}$ appears in $C_3$.
    Let $(H^x,p^x_{C_1},p^x_{C_2},n^x_{C_3})$ be the variable gadget of $x$ constructed as in \Cref{constr:VariableGadgetReduction}.
    Denote by $e^x_{C_1}, e^x_{C_2}$ the positive input edges in $H^x$ adjacent to $p^x_{C_1}, p^x_{C_2}$, respectively, and by $e^x_{C_3}$ the negative input edge in $H^x$ adjacent to $n^x_{C_3}$.
    For any clause $C\in \cC$ and any variable $x\in X$ such that either $x\in C$ or $\bar{x} \in C$, denote by $v^x_C$ the input vertex of $H^x$ which is adjacent to $v_C$ (that is, $v^x_C = p^x_C$ if $x\in C$ and $v^x_C = n^x_C$ if $\bar{x}\in C$).
    
    We start by partitioning into $\{\tilde{F}, \tilde{M} \}$ the edges adjacent to input vertices of variable gadgets, or to vertices in $\{v_C : C \in \cC \}$.
    We then show that this partition can be extended to a valid \mdec{k} $(F_k,M)$ of $G_{(X,\cC)}$ such that $\tilde{F}_k \subset F_k$ and $\tilde{M} \subset M$.
    For each clause $C \in \cC$ we pick one arbitrary variable $x_C$ such that either $x_C$ appears in $C$ and $\phi(x_C)=\true$ or $\bar{x_C}$ appears in $C$ and $\phi(x_C)=\false$ (i.e., an arbitrary variable that makes $C$ evaluate $\true$).
    Then, for each clause $C \in \cC$, we assign $v_Cv^{x_C}_C \in \tilde{M}$ and, consequently, $e^{x_C}_C \in \tilde{F}_k$.
    For any other variable $y \neq x_C$ such that either $y$ or $\bar{y}$ appears in $C$, we assign $v_Cv^y_C \in \tilde{F}_k$ and $e^y_C \in \tilde{M}$.
    For each clause $C \in \cC$ with $|C|=2$ consider the \ffor{1}{k} $(G_C,v_C)$ adjacent to it.
    Denote by $e_C \in E(G_C)$ the edge containing $v_C$, and assign $e_C \in \tilde{F}_k$.

    Note that so far the set of edges which have been partitioned into $\{\tilde{F}_k, \tilde{M} \}$ is precisely the set of all edges adjacent to $v_C$, for all $C \in \cC$, and all input edges (both positive and negative) of all variable gadgets.
    Hence, in order to extend $\{\tilde{F}_k, \tilde{M} \}$ to a \mdec{k} of $G_{(X,\cC)}$, we extend it to a \mdec{k} of each variable gadget and each \ffor{1}{k} attached to a clause of size $2$, separately.

    Starting with the latter, for each \ffor{1}{k} $(G_C,v_C)$ where $C \in \cC$ has $|C|=2$, we fix a \mdec{k} $(F_k^C, M^C)$ of $G_C$ such that $e_C \in F_k^C$.
    By \Cref{def:ellFforcer} and \Cref{prop:ellFforcer} there exists a \mdec{k} satisfies the above.
    
    The set of edges which we need to decompose at this point is precisely
    \begin{align*}
        \bigcup_{x\in X} \left(E(H^x) \setminus \{e^x_{C_1}, e^x_{C_2}, e^x_{C_3} ~:~ x\in C_1, C_2,\, \bar{x} \in C_3 \} \right),
    \end{align*}
    and we do this by extending $\{\tilde{F}_k, \tilde{M})$ to a \mdec{k} on each variable gadget separately.

    Let $x \in X$ and let $C_1, C_2, C_3 \in \cC$ be such that $x$ appears in $C_1, C_2$ and $\bar{x}$ appears in $C_3$.
    Recall that by our choice of $\tilde{F}_k$, having $e^x_{C_3} \in \tilde{F}_k$ implies $\phi(x) = \false$, and having either $e^x_{C_1} \in \tilde{F}_k$ or $e^x_{C_2} \in \tilde{F}_k$ implies $\phi(x) = \true$.
    Hence, neither $\{e^x_{C_1}, e^x_{C_3} \}$ nor $\{e^x_{C_2}, e^x_{C_3} \}$ is contained in $\tilde{F}_k^x$.
    Hence, by \Cref{lem:FillingVariableGadget}, there exists a \mdec{k} $(F_k^x, M^x)$ of $H^x$ for which $\tilde{F}_k^x \subset F_k^x$ and $\tilde{M}^x \subset M^x$, where $\{\tilde{F}_k^x, \tilde{M}^x \}$ is the partition of $\{e^x_{C_1}, e^x_{C_2}, e^x_{C_3} \}$ induced by $\{\tilde{F}_k, \tilde{M} \}$.

    Define $(F_k,M)$ by
    \begin{align*}
        F_k &\coloneqq \left(\bigcup_{x \in X} F_k^x \right) \cup \left(\bigcup_{|C|=2} F_k^C \right) \\
        M &\coloneqq \left(\bigcup_{x \in X} M^x \right) \cup \left(\bigcup_{|C|=2} M^C \right).
    \end{align*}

    It is left to show that the decomposition $(F_k,M)$ is a valid \mdec{k} of $G_{(X,\cC)}$, meaning that $M$ is indeed a matching and $F_k$ is indeed a $k$-bounded linear forest.
    Note that, by our construction, for each $x \in X$ we have that $M^x$ and $F_k^x$ are indeed a matching and a $k$-bounded linear forest, respectively.
    The same holds for $(F_k^C,M^C)$ for each $C \in \cC$ with $|C|=2$.
    Since each input vertex (either positive or negative) in each variable gadget $H^x$ is contained in exactly one edge in $M$, it follows that $M$ is indeed a matching in $G_{(X, \cC)}$.
    
    Suppose for the sake of contradiction that there exists a subgraph $Q$ of the graph spanned by the edges in $F_k$ which is either a path of length $k+1$ or a cycle, or contains a vertex of degree $3$.
    Note that by \Cref{lem:FillingVariableGadget} and \Cref{def:ellFforcer}, $Q$ cannot be fully contained in any variable gadget or in any \ffor{1}{k}.
    Hence, there exists a clause $C \in \cC$ such that $v_C \in V(Q)$.
    As by our construction we have $v_Cv^{x_C}_C \in M$ and for any other variable $y \neq x_C$ appearing in $C$, we have $e^y_C \in M$, it follows that $Q$ has size at most $2$, which is a contradiction.
    Hence $F_k$ is a $k$-bounded linear forest and this finishes the proof.
\end{proof}

\begin{lemma}
\label{lem:MatchingForestto3-SAT}
    Let $k \ge 3$.
    Let $(X,\cC)$ be an instance of Planar $(\le 3,3)$–SAT, and let the graph $G_{(X,\cC)}$ be as in \Cref{constr:GraphReductionNP}.
    If $G_{(X,\cC)}$ is a yes-instance of Planar \blfd{k}{1}, then $(X,\cC)$ is a yes-instance of Planar $(\le 3,3)$–SAT.
\end{lemma}

\begin{proof}
Let $(F_k,M)$ be a decomposition of $G_{(X,\cC)}$ such that $F_k$ is a $k$-bounded linear forest and $M$ is a matching.
We construct a truth assignment $\phi : X \rightarrow \left\{\true, \false \right\}$.
Let $x$ be a variable in $X$, and suppose that $x$ appears in clauses $C_1$ and $C_2$ and that $\bar{x}$ appears in $C'$.
By \Cref{lem:FillingVariableGadget}, we have neither $e^x_{C_1}, e^x_{C'} \in F_k$ nor $e^x_{C_2}, e^x_{C'} \in F_k$.
We define 
\begin{align*}
    \phi(x) = \begin{cases}
    \true & \text{ if } e^x_{C_1} \in F_k \text{ or } e^x_{C_2} \in F_k\\
    \false & \text{ otherwise}.
    \end{cases}
\end{align*}
We claim that every clause $C \in \cC$ has at least one true literal.
Indeed, $v_C$ has degree 3 in $G_{(X,\cC)}$ so there exists $x \in X$ such that $v_Cv^x_C \in M$, hence $e^x_{C} \in F_k$.
If $v^x_C = p^x_C$, then $x$ appears in $C$ and $\phi(x)=\true$ by our construction.
If $v^x_C = n^x_C$, then $\bar{x}$ appears in $C$ and then $\phi(x)=\false$.
Hence, $\phi$ is a satisfying assignment for $(X,\cC)$ and consequently $(X,\cC)$ is a yes-instance of Planar $(\le 3,3)$–SAT.
\end{proof}

\Cref{thm:main} then follows from \Cref{thm:NPcomplete3bounded3SAT}, \Cref{lem:3-SATtoMatchingForest} and \Cref{lem:MatchingForestto3-SAT}.

\section{Decomposing a graph into a matching and a (bounded) star forest}\label{sec:StarForestinP}

As mentioned in the introduction, by generalisation of the proof in~\cite{campbell2023decompositions} of \blfd{2}{1} being polynomial time solvable, one can prove \Cref{thm:DecompositionMatchingStarForest}. More precisely, we do so using a polynomial reduction to the \emph{small gap general factor problem}, which is solvable in polynomial time by a result of Cornuéjols~\cite{cornuejols1988general}. 

We say that a set $A \subset \mathbb{N}$ is a \emph{small gap set} if for every integer $i \in \mathbb{N}$ satisfying $\inf{A} \leq i \leq \sup{A}$, we have either $i \in A$ or $i+1 \in A$.
Given a graph $G = (V,E)$, a subset $S \subset E$ of its edges, and a vertex $v \in V$, we denote by $d_S(v)$ the degree of $v$ in the subgraph of $G$ spanned by the edges in $S$.
The small gap general factor problem is defined as follows.

\subsubsection*{Small gap general factor problem (SGGF)}
\textbf{Input:} 
A graph $H$, a collection of small gap sets $\{A_v ~:~ v \in V(H) \}$. \\
\textbf{Question:} Is there a set $S \subseteq E(H)$ such that $d_S(v) \in A_v$ for every $v \in V(H)$? \\

As mentioned earlier, Cornuéjols~\cite{cornuejols1988general} proved the following.

\begin{theorem}
\label{thm:SmallGapisP}
    SGGF is solvable in polynomial time.
\end{theorem}

\subsection{The reduction}

We start by an observation that was already noticed in \cite{campbell2023decompositions} for the case $k=2$.

\begin{observation}
Let $G$ be a graph containing a vertex of degree at least $k+2$.
Then $G$ is a no-instance of $k$-MBSFD.
\end{observation}

We now present the reduction, which is a simple generalisation of the reduction used for the case $k=2$ proved in~\cite{campbell2023decompositions}.

\begin{construction} 
\label{constr:ReductionStarForestSmallGap}
Let $k \in \{2, \ldots \} \cup \{\infty \}$ and let $G$ be an instance of $k$-MBSFD.
We construct the following instance $\left(H_G, \{ A_v : v \in V(H_G) \} \right)$ of SGGF.
For an integer $d$ let $X_{\ge d}, X_d \subset V(G)$ be the sets of vertices which have degree at least $d$ and exactly $d$ in $G$, respectively.
Let $\cP$ be the family of all paths in $G$ with both endpoints in $X_{\ge 3} \cup X_1$, and with all interior vertices in $X_2$, together with all cycles in $G$ containing at most one vertex in $X_{\ge 3}$.
For an integer $\ell \ge 1$ let $\cP_{\ell} \subset \cP$ be the subfamily of all such paths and cycles of length $\ell$.
Let $H=H_G$ be a bipartite graph on parts $X=X_{\ge 3}$ and $Y= \left\{ y_P : P \in \cP\right\}$, and with $E(H) = \{ xy_P : x \text{ is an endpoint of } P \text{ in } G \}$.

We define a small gap set $A_v$ for each vertex $v \in V(H)$.
\begin{itemize}
    \item For every $x \in X_{\ge k+1}$, we set $A_x=\{1\}$.
    \item For every $x \in X \setminus X_{\ge k+1}$, we set $A_x=\{0,1\}$.
    \item For every $P \in \cP_1$, we set $A_{y_P}=\{ 2 \}$.
    \item For every $P \in \cP_2$, we set $A_{y_P}= \{ 1 \}$.
    \item For every $P \in \cP_3$, we set $A_{y_P}= \{ 0,2 \}$.
    \item For every $P \in \cP_4$, we set $A_{y_P}= \{ 1,2 \}$.
    \item For every $P \in \cP \setminus (\cP_1 \cup \cP_2 \cup \cP_3 \cup \cP_4)$, we set $A_{y_P}= \{0,1,2 \}$.
\end{itemize}
\end{construction}

\begin{remark}
Note that for $k=\infty$ \Cref{constr:ReductionStarForestSmallGap} still holds, but since $X_{\ge k+1} = \emptyset$ the definition of small gap sets in the first item is null.
In addition, for $k=2$ \Cref{constr:ReductionStarForestSmallGap} is identical to the one in~\cite{campbell2023decompositions}.
\end{remark}

The following two lemmas together imply that this reduction is valid.
Their proofs are essentially the same as the proofs of their equivalent lemma for the case $k=2$ in~\cite{campbell2023decompositions}, following the exact same casework, except for one difference.
Hence, we do not include here the full proofs.
However, for the sake of clarity, we give the general idea of their proofs with the details of the parts which we generalise.

\begin{lemma}
\label{lem:SmallGapToMBSFD}
Let $k\ge 2$ be an integer and $G$ be a graph. Let $\left(H_G, \{ A_v : v \in V(H_G) \} \right)$ be as in \Cref{constr:ReductionStarForestSmallGap}. If $\left(H_G, \{ A_v : v \in V(H_G) \} \right)$ is a yes-instance of SGGF, then $G$ is a yes-instance of $k$-MBSFD.
\end{lemma}

\begin{lemma}
\label{lem:MBSFDToSmallGap}
Let $k\ge 2$ be an integer and $G$ be a graph. Let $\left(H_G, \{ A_v : v \in V(H_G) \} \right)$ be as constructed in \Cref{constr:ReductionStarForestSmallGap}.
If $G$ is a yes-instance of $k$-MBSFD, then $\left(H_G, \{ A_v : v \in V(H_G) \} \right)$ is a yes-instance of SGGF.
\end{lemma}

\Cref{thm:DecompositionMatchingStarForest} then follows from \Cref{lem:SmallGapToMBSFD}, \Cref{lem:MBSFDToSmallGap} and \Cref{thm:SmallGapisP}.

\begin{proof}[Proof ideas for \Cref{lem:SmallGapToMBSFD} and \Cref{lem:MBSFDToSmallGap}]
    Let $G$ be a graph and $\left(H_G, \{A_v : v\in V(H_G) \} \right)$ be an instance of SGGF as in \Cref{constr:ReductionStarForestSmallGap}.
    The basic idea of the reduction relies on an observation that a valid decomposition of $G$ into a matching $M$ and a $k$-bounded star forest $S_k$ is equivalent to a decomposition of the edges adjacent to vertices of degree at least $3$, in which the edges in $M$ are a solution to $\left(H_G, \{A_v : v\in V(H_G)\} \right)$.
    In other words, choosing the matching edges amongst the edges of the stars with centres at $X_{\ge 3}$, such that it satisfies the conditions for the elements in $\cP$ and in $X_{\ge 3}$ as given by $\{A_v : v\in V(H_G)\}$, is sufficient and necessary for constructing a decomposition of $G$ into a matching and a $k$-bounded star-forest.
    Hence, $G$ being a yes-instance of $k$-MBSFD implies $\left(H_G, \{A_v : v\in V(H_G) \} \right)$ being a yes-instance of SGGF, and vice-versa.

    By looking at the choice of $\{A_v : v \in V(H_G) \}$ one can see that vertices $x \in X_{\ge k+1}$ of $G$ must be adjacent to a matching edge, and vertices $x \in X_{\ge 3} \setminus X_{\ge k+1}$ in $G$ must be adjacent to at most one matching edge.
    If $k=2$ the latter is clearly an emptyset, which demonstrate the only difference in our construction.
    The proof of the reduction in~\cite{campbell2023decompositions} follow a casework which shows the above correspondence between a solution $S$ to the SGGF problem, and a set of matching edges chosen from the edges adjacent to vertices in $X_{\ge 3}$.
    Separating the case of vertices in $X_{\ge 3}$ into two cases, as in \Cref{constr:ReductionStarForestSmallGap}, $X_{\ge k+1}$ and $X_{\ge 3} \setminus X_{\ge k+1}$, adds one more case to check, which is straightforward.
    This, in fact, is the only difference between the proof of our reduction and the one in~\cite{campbell2023decompositions}.
\end{proof}

\section{Concluding remarks and open problems}\label{sec:ConcludingRemarks}

While our constructions in Section \ref{sec:NPComplete} require $k$ to be finite, we remark that it is possible to alter Campbell, H{\"o}rsch and Moore's \cite{campbell2023decompositions} construction into proving that Planar \blfd{\infty}{1} is NP-complete.
The only two differences would be to reduce from Planar $(\leq 3, 3)$-SAT rather than from $(3,B2)$-SAT, and to use simple graph \mfor{8} and \ffor{1}{8} from our paper instead of the multigraph forcers from \cite{campbell2023decompositions}.

As the complexity of \blfd{k}{l} has now been fully characterized, we would like to direct the reader to recent open questions concerning the existence of certain decompositions of graphs into linear forests and a matching.
The first comes from the paper of Campbell, H{\"o}rsch and Moore \cite{campbell2023decompositions} and still concerns subcubic planar graphs:

\begin{problem}[Problem $4$ in \cite{campbell2023decompositions}]
What is the minimum integer $\alpha$ such that every planar subcubic graph of girth at least $\alpha$ can be decomposed into a linear forest and a matching?
\end{problem}

So far it is known only that $6 \leq \alpha \leq 9$. As for a subcubic graph, a decomposition into a matching and a linear forest is stronger than a $3$-edge colouring, then note that improving $\alpha$ to something smaller than $9$ would imply results on $3$-edge colourings (see, for instance \cite{bonduelle22edgecolourings,kronk1974line}).

The other open problem is related to the Linear Arboricity Conjecture. Bonamy, Czy{\.z}ewska, Kowalik and Pilipczuk~\cite{bonamy2023partitioning} studied the cases when $\Delta(G)$ when $G$ is a planar graph of odd maximum degree, searching for even stronger conditions on the $\frac{\Delta(G)+1}{2}$ linear forests $G$ decomposes into. They conjectured that any planar graph $G$ of odd maximum degree $\Delta \ge 7$ can be decomposed into $\frac{\Delta-1}{2}$ many linear forests and a matching. They proved their conjecture for odd $\Delta \geq 9$, leaving only one case open:

\begin{conjecture}[Conjecture $2$ in \cite{bonamy2023partitioning}]
Any planar graph $G$ of maximum degree $7$ can be decomposed into $3$ linear forests and one matching.
\end{conjecture}

\section*{Acknowledgement}

The first, third and fifth authors would like to thank IMPA for their hospitality during a visit in which this work started.
The authors are thankful to Professor B\'{e}la Bollob\'{a}s and Professor Rob Morris, who form the set of the authors' supervisors, for their valuable input. 

The second author is partially supported by the Coordenação de Aperfeiçoamento de Pessoal de Nível Superior, Brasil (CAPES).
The third author is supported by Trinity College, Cambridge.
The last two authors are supported by EPSRC (Engineering and Physical Sciences Research Council).
The fourth author is also supported by the Department of Pure Mathematics and Mathematical Statistics of the University of Cambridge, and the fifth author is also supported by the Cambridge Commonwealth, European and International Trust.

\bibliographystyle{abbrvnat}  
\renewcommand{\bibname}{Bibliography}
\bibliography{main}

\end{document}